\documentclass{article}

\usepackage[accepted]{icml2018}

\usepackage{microtype}
\usepackage{graphicx}
\usepackage{booktabs} 

\usepackage{amsmath}
\usepackage{graphicx}
\usepackage{amssymb}
\usepackage{amsthm}
\usepackage[noend]{algpseudocode}
\usepackage{tabularx}
\usepackage{blindtext}
\usepackage{mathtools}
\usepackage{subfig}
\usepackage{url}

\usepackage{graphicx}
\usepackage{caption}
\usepackage[utf8]{inputenc}
\usepackage[export]{adjustbox}
\usepackage{wrapfig}
\usepackage{float}

\makeatletter
\providecommand\citet[2][]{%
  \edef\@tempa{#1}
  \citeauthor{#2} (%
    \citeyear{#2}%
    \ifx\@empty\@tempa\else,~#1\fi
  )
}
\makeatother

\newcommand{\eat}[1]{}
\newcommand{\R}{\mathbb{R}}

\theoremstyle{definition}\newtheorem{thm}{Theorem}[section]
\theoremstyle{definition}\newtheorem{mydef}[thm]{Definition}
\theoremstyle{definition}\newtheorem{rem}[thm]{Remark}
\theoremstyle{definition}\newtheorem{prop}[thm]{Proposition}
\theoremstyle{definition}
\theoremstyle{definition}
\theoremstyle{definition}
\theoremstyle{definition}\newtheorem{Lemma}[thm]{Lemma}
\theoremstyle{definition}
\theoremstyle{definition}\newtheorem{Fact}[]{Fact}

\DeclareMathOperator*{\argmin}{argmin}
\algtext*{EndProcedure}%

\newenvironment{myproof}[3]{
    \begin{proof}[\normalfont \textbf{Proof of #1 #2.}]
        #3
       \end{proof}
}

\icmltitlerunning{Leveraging Well-Conditioned Bases: Streaming \& Distributed Summaries in Minkowski $p$-Norms}

\begin{document}

\twocolumn[
\icmltitle{Leveraging Well-Conditioned Bases: Streaming and \\
        Distributed Summaries in Minkowski $p$-Norms}



\icmlsetsymbol{equal}{*}

\begin{icmlauthorlist}
\icmlauthor{Graham Cormode}{equal,warwick}
\icmlauthor{Charlie Dickens}{equal,warwick}
\icmlauthor{David P. Woodruff}{equal,cmu}

\end{icmlauthorlist}

\icmlaffiliation{warwick}{Department of Computer Science, University of Warwick, Coventry, UK}
\icmlaffiliation{cmu}{School of Computer Science, Carnegie Mellon University, Pittsburgh, Pennsylvania, USA}

\icmlcorrespondingauthor{Charlie Dickens}{c.dickens@warwick.ac.uk}

\icmlkeywords{Machine Learning, ICML}

\vskip 0.3in
]



\printAffiliationsAndNotice{\icmlEqualContribution} 

\begin{abstract}
  Work on approximate linear algebra
  has led to efficient distributed and streaming
  algorithms for
  problems such as approximate matrix multiplication, low rank approximation,
  and regression, primarily for the Euclidean norm $\ell_2$.
  We study other
  $\ell_p$ norms, which are more robust for $p < 2$, and can be used
  to find outliers for $p > 2$.
  Unlike previous algorithms for such norms,
we give algorithms that are (1) deterministic, (2) work simultaneously
for every $p \geq 1$, including $p = \infty$, and (3) can be
  implemented in both
  distributed and streaming environments. We apply our results to $\ell_p$-regression,
  entrywise $\ell_1$-low rank approximation,
  and approximate matrix multiplication.
\end{abstract}

\section{Introduction}

Analyzing high dimensional, high volume data can be time-consuming and
resource intensive.
Core data analysis, such as robust instances of regression, involve convex optimization tasks over large matrices, and do
not naturally distribute or parallelize.
In response to this, approximation algorithms have been proposed which
follow a ``sketch and solve'' paradigm: produce a reduced size
representation of the data, and solve a version of the problem on this
summary~\cite{Woodruff:15}.
It is then argued that the solution on the reduced data provides an
approximation to the original problem on the original data.
This paradigm is particularly attractive when the summarization can be
computed efficiently on partial views of the full data---for
example, when it can be computed incrementally as the data arrives
(streaming model) or assembled from summarizations of disjoint
partitions of the data (distributed model)~\cite{Woodruff:15,Agarwal:Cormode:Huang:Phillips:Wei:Yi:12,Feldman06}.
This template has been instantiated for a number of
fundamental tasks in high dimensional linear algebra such as matrix
multiplication, low rank approximation, and regression.

Our understanding is well-established in the common
case of the Euclidean norm, i.e., when distances are measured under the
Minkowski $p$-norm for $p=2$.
Here, it suffices to choose a sketching matrix independent of the
data---where each entry is i.i.d. Gaussian, Rademacher, or 
more efficient variants of these. 
For other $p$ values, less is known, but these are often needed to
handle limitations of the $2$-norm.
For instance, $p=1$ is widely used as it is extremely robust with respect to
the presence of outliers while $p>2$ can be used to detect outlying observations.

We continue the study of algorithms for $\ell_p$ norms 
on streaming and distributed data.
%
%
A particular novelty of our results is that unlike previous distributed and
streaming algorithms, 
they can all be implemented
{\it deterministically}, i.e., our algorithms 
make no random choices.
While in a number of settings randomized algorithms are highly beneficial,
leading to massive computational savings, there are other applications
which require extremely high reliability, 
for which one needs to obtain guaranteed performance
across a large number of inputs. 
If one were to use a randomized algorithm, then it would need
vanishingly small error probability;
however, many celebrated algorithms in numerical linear algebra succeed with
only constant probability. 
Another limitation of randomized algorithms 
was shown in \cite{Hardt:2013:RLS:2488608.2488624}:
if the input to a randomized sketch depends on the output of a preceding
algorithm using the same sketch, then the randomized sketch can give an
arbitrarily bad answer. Hence, such methods cannot handle adaptively chosen
inputs. 
Thus, while randomized algorithms certainly have their place, 
the issues of high reliability and adaptivity motivate the development of deterministic methods for a number
of other settings, for which algorithms are scarce. 

Our techniques can be viewed as a conceptual generalization of Liberty's
Frequent Directions (in the $2$-norm)~\cite{L13}, which progressively computes
an SVD on subsequent blocks of the input.
This line of work \cite{L13,gp14,G16,GLPW16} is the notable exception
in numerical linear algebra, as it provides deterministic methods, although all such methods are specific to the 2-norm.
Our core algorithm is similar in nature, but we require a very
different technical analysis to argue that the basis transformation
computed preserves the shape in the target $p$-norm.

Our main application is to show how high dimensional regression and low rank approximation
problems
can be solved approximately and deterministically 
in the sketch and solve paradigm.
The core of the summary is to find rows of the original matrix which
have high {\em leverage scores}.
That is, they contain a lot of information about the shape of the
data.
In the Euclidean norm,
leverage scores correspond directly to row norms of an orthonormal basis.
This is less
straightforward for other $\ell_p$ norms, where the scores correspond to
the row norms of so-called $\ell_p$-well-conditioned bases. Moreover, while leverage
scores are often used for sampling in randomized algorithms, we use them here in the
context of fully deterministic algorithms. 

We show how a superset of rows with high leverage scores can be found
for arbitrary $\ell_p$ norms, based on only local information.
This leads to efficient algorithms which identify rows with high
(local) leverage scores within subsets of the data, and proceed
hierarchically to collect a sufficient set of rows.
These rows then allow us to solve regression problems: essentially, we
solve the regression problem corresponding to just the retained
input rows.
We apply this technique to $\ell_p$-regression and entrywise $\ell_p$-low rank approximation.
In particular, we use it to solve the $\ell_\infty$-regression problem with additive error in a
stream.
Note that the $\ell_{\infty}$ problem reduces to finding a ball of minimum radius which
covers the data, and global solutions are slow due to the need to solve a linear program. 
Instead, we show that only a subset of the data needs to be retained in the streaming model 
to compute accurate approximations. 
Given the relationship between the streaming model and the distributed model that we later define, this could be seen in the 
context of having data stored over multiple machines who could send `important' rows of their
data to a central coordinator in order to compute the approximation.

\textbf{Summary of Results.}
All our algorithms are deterministic polynomial time, and use significantly
sublinear memory or communication in streaming and distributed models,
respectively. We consider
tall and thin $n \times d$ matrices $A$ for overconstrained
regression 
so one should think of $n \gg d$. We 
implement both deterministic and randomized variants of our algorithms.
\newline
\textbf{Section \ref{HighLevScores}} presents an algorithm which
returns rows of high `importance' in a data matrix with additive
error.  This follows by storing a polynomial number (in $d$) of rows
and using these to compute a well-conditioned basis.  The key insight
here is that rows of high norm in the full well-conditioned basis
cannot have their norm decrease too much in a well-conditioned basis
associated with a subblock; in fact they remain large up to a
multiplicative poly$(d)$ factor.
\newline
\textbf{Section \ref{Deterministic ell+p subspace embedding
       section}} gives a method for computing a so-called $\ell_p$-subspace
   embedding of a data matrix in polynomial time. The space is $n^{\gamma}$
   to obtain $d^{O(1/\gamma)}$ distortion, for $\gamma \in (0,1)$ a small
   constant.
    This result is then applied to $\ell_p$-regression which is shown to have a poly$(d)$ approximation factor with the same amount of space.
\newline
    \textbf{Section \ref{low_rank_main}} describes a deterministic algorithm which gives a poly$(k)$-approximation to the optimal low rank approximation problem in entrywise $\ell_1$-norm. It runs in polynomial time for constant $k$. This method builds on prior work by derandomizing a subroutine from \cite{woodruff_ell_1_low_rank}. 
\newline
\textbf{Section \ref{Deterministic ell_inf regression}} describes an
   algorithm for  computing an additive-error solution to the
   $\ell_{\infty}$-regression problem, and shows a corresponding
   lower bound, showing that relative error solutions in this norm are
   not possible in sublinear space, even for randomized algorithms.
\newline
\textbf{Section} \ref{sec: experiments} concludes with an empirical evaluation. 
   More experiments, intermediate results, and formal proofs
   can be found in the Supplementary Material, as can results on 
   approximate matrix multiplication.


\eat{\textbf{Other Related Work.}
There is a rich literature on randomized algorithms for linear algebra 
for general $p$-norms using row sampling
(\cite{cohen2015p}, \cite{Mahoney}) and low-distortion embeddings
(\cite{fast_cauchy}, \cite{mm13}, \cite{woodruff_ell_1_low_rank},
\cite{Woodruff:Zhang:13}).
 While these can solve instances of
regression in time proportional to the sparsity of the input matrix, they 
are all randomized. Also, several of them do not
easily translate to the distributed and streaming models.
}

\textbf{Comparison to Related Work.}
There is a rich literature on algorithms for numerical linear algebra 
in general $p$-norms; most of which are randomized with the notable 
exception of Frequent Directions.
The key contributions of our work for each of the problems
considered and its relation to prior work is as follows:

\textit{Finding high leverage rows:} our algorithm is a single pass streaming 
    algorithm and uses small space.  
    We show that the global property of $\ell_p$-leverage scores can be understood by 
    considering only local statistics.
    Frequent Directions is the only comparable result to ours and 
    outputs a summary of the rows only in the 
    $\ell_2$-norm.
    However, our method covers all $p \ge 1$.
    Theorem \ref{DetLevScoreThm} is the key result and is later used to prove 
    Theorem \ref{ell_inf reg} and approximate the 
    $\ell_{\infty}$-regression problem.
    
\textit{Subspace embedding, regression and $\ell_1$ low-rank approximation:} 
    various approaches using row-sampling \cite{cohen2015p,Mahoney},
    and data oblivious methods
    such as low-distortion embeddings can solve regression in  time
     proportional to the sparsity of the input matrix \cite{fast_cauchy,mm13,woodruff_ell_1_low_rank,Woodruff:Zhang:13}.
     However, despite the attractive running times and error guarantees of these works,
     they are all randomized and do not necessarily translate well
     to the streaming model of computation.
     Our contribution here is a fully deterministic algorithm that works for all $p \ge 1$ 
     in both streaming and distributed models.
     Randomized methods for $\ell_1$ low-rank approximation have also been developed
     in \cite{woodruff_ell_1_low_rank} and
     our result exploits a derandomized subroutine from this work
    to obtain a deterministic result which applies in both models.

\section{Preliminaries and Notation}
We consider computing $\ell_p$-leverage
scores of a matrix, low-rank approximation, regression, and matrix multiplication.
We assume the input is a matrix $A \in \R^{n \times d}$ and $n \gg d$ so $\operatorname{rank}(A) \le d$ and the regression problems are
overconstrained.
Without loss of generality we may assume that the columns of the input matrix
are linearly independent so that $\operatorname{rank}(A) = d$.  Throughout this
paper we rely heavily on the notion of a \emph{well-conditioned basis} for the
column space of an input matrix, in the context of the
\emph{entrywise $p$-norm} which is $\| A \|_p = (\sum_{i,j} |A_{ij}|^p)^{1/p}$.
\begin{mydef}[Well-conditioned basis] \label{WCB}
Let $A \in \R^{n \times d}$ have rank $d$.
For $p \in [1, \infty)$ let $q = \frac{p}{p-1}$ be its dual norm. An $ n \times d$ matrix $U$ is an $(\alpha, \beta, p)$-well-conditioned basis for $A$ if the column span of $U$ is equal to that of $A$, $\Vert U \Vert_p \le \alpha$, \label{prop1} for all $z \in \R^d, \Vert z \Vert_q \le \beta \Vert U z \Vert_p$ \label{prop2}, and $\alpha, \beta, d^{O(1)}$ are independent of $n$ \cite{Mahoney}.

\eat{We say that $U$ is a $\ell_p$-well-conditioned basis for the
column space of $A$ if $\alpha$ and $\beta$ are $d^{O(1)}$, independent of $n$ \cite{Mahoney}.}
\end{mydef}

We focus on the cases $p < 2$ and $p >2$ because the deterministic $p
= 2$ case is relatively straightforward.
Indeed, for $p = 2$, $A^TA$ can be maintained incrementally
as rows are added, allowing $x^TA^TAx$ to be computed for any vector $x$.  So it is possible to find an exact $\ell_2$ subspace embedding using $O(d^2)$ space in a stream and $O(nd^{\omega -1})$ time ($\omega$ is the matrix multiplication constant).
We adopt the convention that when $p=1$ we take $q = \infty$.

\begin{thm}[\cite{Mahoney}] \label{WCBthm}
   Let $A$ be an $n \times d$ matrix of rank $d$, let $p \in [1, \infty)$ and let $q$ be its dual norm.  There exists an $(\alpha, \beta, p)$-well-conditioned basis $U$ for the column space of $A$ such that:
\begin{enumerate}
\addtolength{\itemsep}{-1ex}
\item if $p<2$ then $\alpha = d^{\frac{1}{p} + \frac{1}{2}}$ and $\beta = 1$,
\item if $p = 2$ then $\alpha = \sqrt{d}$ and $\beta = 1$, and
\item if $p > 2$ then $\alpha = d^{\frac{1}{p} + \frac{1}{2}}$ and $\beta = d^{\frac{1}{p} - \frac{1}{2}}$.
\end{enumerate}
Moreover, $U$ can be computed in deterministic time $O(nd^2 + n d^5 \log n)$ for $p \ne 2$ and $O(nd^2)$ if $p = 2$.
\end{thm}

We freely use the fact that a well-conditioned basis $U = AR$ can be efficiently computed for the given data matrix $A$.
Details for the
computation can be found in \cite{Mahoney} but this is done by computing a change of basis $R$ such that $U = AR$ is well-conditioned.
Similarly, as $R$ can be inverted we have the relation that $UR^{-1} = A$.
Both methods are used so we adopt the convention that $U=AR$ when writing a well-conditioned basis in terms of the input and $US = A$ for the input in terms of the basis.
\subsection{Computation Models} \label{sec: models}

Our algorithms operate under the \textit{streaming} and \textit{distributed} models of computation.
In both settings an algorithm receives as input a matrix $A \in \R^{n \times d}$.
For a problem \textbf{P}, the algorithm must keep a subset of the rows of $A$ and, upon reading the full input, may use a black-box solver to compute an approximate solution to \textbf{P} with only the subset of rows stored.
In both models we measure the \textit{summary size} (storage), the \textit{update time} which is the time taken to find the local summary, and the \textit{query time} which is the time taken to compute an approximation to \textbf{P} using the summary.

\textbf{The Streaming Model}:
The rows of $A$ are given to the (centralized) algorithm one-by-one.
Let $b$ be the maximum number of rows that can be stored under the constraint that $b$ is sublinear in $n$.
The stored subset is used to compute \textit{local} statistics which determine those rows to be kept or discarded from the stored set.
Further rows are then appended and the process is repeated until the full matrix has been read.
An approximation to the problem is then computed by solving \textbf{P} on the reduced subset of rows.

\textbf{The Distributed Summary Model}:
Given a small constant $\gamma \in (0,1)$,
the input in the form of  matrix $A \in \R^{n \times d}$ 
%
is partitioned into blocks among distributed compute nodes so
that no block exceeds $n^{\gamma}$ rows.
The computation then follows a tree structure: 
the initial blocks of the matrix form $n^{1-\gamma}$ leaves of the
compute tree.
Each internal node merges and reduces its input from its child nodes.
The first phase is for the leaf nodes $l_1, \ldots, \l_m$ of the tree to reduce their input by computing a local summary on the block they receive as input.
This is then sent to parent nodes $p_1, \ldots, p_m$ which
 merge and reduce the received rows until the space bound is reached.
The resulting summaries are passed up the tree until we reach the
root where a single summary of bounded size is obtained which can be
used to compute an approximation to \textbf{P}.  
In total, there are $O(1/\gamma)$ levels in the tree. 
As the methods require only light synchronization (compute summary and
return to coordinator), we do not model implementation issues
relating to synchronization. 

\begin{rem} \label{rem: tree-stream}
The two models are quite close: 
the \textit{streaming model} can be seen as a special case of the
distributed model with only one participant who individually computes
a summary, appends rows to the stored set, and reduces the new
summary.
This is represented as a deep binary tree, where each internal node
has one leaf child. 
Likewise, the \textit{Distributed Summary Model} can be implemented in
a full streaming fashion over the entire binary tree.
The experiments in Section \ref{sec: experiments} perform one round of merge-and-reduce in the distributed model to simulate the streaming approach.
\end{rem}

\section{Finding Rows of High Leverage} \label{HighLevScores}

This section is concerned with finding rows of high leverage from a matrix with respect to various $p$-norms. We conclude the section with an algorithm that returns rows of high leverage up to polynomial additive error.

\begin{mydef}
Let $R$ be a change of basis matrix such that $AR$ is a well-conditioned basis
for the column space of $A$.
The (full) $\ell_p$-leverage scores are defined
as $w_i = \| \mathbf{e}_i^T AR \|_p^p$.
\end{mydef}

\noindent
Note that $w_i$ depends both on $A$ and the choice of $R$, but we
suppress this dependence in our notation.
Next we present some basic facts about the $\ell_p$ leverage scores.

\begin{Fact} \label{LevBound1}
By Definition \ref{WCB} we have $\sum_i w_i = \sum_i \Vert (AR)_i \Vert_p^p \le \alpha^p.$
Theorem \ref{WCBthm} shows $\alpha = \text{poly}(d)$.  Define $I = \{i
\in [n] : w_i > \tau \| AR \|_p^p\}$ to be the index set of all rows
whose $\ell_p$ leverage exceeds a $\tau$ fraction of $\| AR \|_p^p$,
then:
$\alpha^p \ge \sum_i w_i \ge \sum_{i \in I} w_i \ge |I| \cdot \tau \| AR \|_p^p.$
Hence, $|I| \le \alpha^p/\tau \| AR \|_p^p = \text{poly}(d)/\tau$.  So there are at most $\text{poly}(d) / \tau$ rows $i$ for which $w_i \ge \tau \| AR \|_p^p$.
\end{Fact}

\begin{Fact} \label{LevBound2}
Definition \ref{WCB} and H\"{o}lder's inequality show that for any vector $\mathbf{x}$ we have $ |(AR\mathbf{x})_i  |^p \le \beta \Vert \mathbf{e}_i^T AR \Vert_p^p \cdot \Vert ARx \Vert_p^p.$
Then $\tau \le | \mathbf{e}_i^T AR \mathbf{x}|^p/ \Vert AR \mathbf{x} \Vert_p^p \le \beta w_i.$
From this we deduce that if a row contributes at least a $\tau$ fraction of $ \Vert AR \mathbf{x} \Vert_p^p $ then $\tau \le w_i \beta$.  That is, $\tau \le w_i$ for $p \in [1,2]$ and $\tau \le d^{1/2} w_i$ for $p \in (2, \infty)$ by using Theorem \ref{WCBthm}.
\end{Fact}

\begin{mydef}
Let $X$ be a matrix and $Y$ be a subset of the rows of $X$. Define the \emph{local $\ell_p$-leverage scores of $Y$ with respect to $X$} to be the leverage scores of rows $Y$ found by computing a well-conditioned basis for $Y$ rather than the whole matrix $X$.
\end{mydef}
A key technical insight to proving Theorem \ref{DetLevScoreThm} below is that rows
of high leverage globally can be found by repeatedly finding rows of local high
leverage.
While relative $\ell_p$ row norms of a submatrix are at least as
large as the full relative $\ell_p$ norms, it is not guaranteed that
this property holds for leverage scores.  This is because leverage
scores are calculated from a well-conditioned basis for a matrix which
need not be a well-conditioned basis for a block.  However, we show
that local $\ell_p$ leverage scores restricted to a coordinate subspace of a
matrix basis do not decrease too much when compared to leverage scores in the
original space. 
Let $i$ be a row in $A$ with local leverage score $\hat{w}_i$ and global
leverage score $w_i$.
Then $\hat{w}_i \ge w_i / \operatorname{poly}(d)$.
The proof relies heavily on properties of the well-conditioned basis and details
are given in the Supplementary Material, Lemma  \ref{lem: lev_scores_drop}.
This lemma shows that local leverage scores can potentially drop in
arbitrary $\ell_p$ norm, contrasting the behavior in $\ell_2$.
However, it is possible
to find all rows exceeding a threshold globally by altering the local threshold.
That is, to find all $w_i > \tau$ globally we can find all local leverage scores
exceeding an adjusted threshold $\hat{w}_i > \tau / \operatorname{poly}(d)$
to obtain a superset of all rows which exceed the global threshold.
The price to pay for this is a $\operatorname{poly}(d)$ increase in space cost
which, importantly, remains \textit{sublinear in} $n$.
Hence, we can gradually prune out rows of small leverage and keep only the most
important rows of a matrix.
Combining Lemmas \ref{lem: lev_scores_drop} and \ref{lem: space_local_vs_global}
 we can present the main theorem of the section.


%


\begin{figure*}[ttt!]
 \begin{minipage}[t]{0.6\textwidth}
 \begin{algorithm}[H]
\caption{Deterministic High Leverage Scores} \label{DetLevScoreAlg}
\begin{algorithmic}[1]
\Require $A \in \R^{n \times d}, \tau \in (0,1)$
\Procedure{High Leverage Scores}{$A, \tau$}
\State $b \leftarrow \text{poly}(d)/ \tau$
\State $A' \leftarrow$ first $b$ rows of $A$
\State $B \leftarrow \textsc{LevScoreCheck}(\text{wcb}(A'), A', \tau/\text{poly}(d))$
\While{Rows of $A$ unseen}
    \State $A' \leftarrow $  next $b$ rows of $A$
    \State $B \leftarrow \textsc{LevScoreCheck}(\text{wcb}([A'; B]), [A'; B], \tau/\text{poly}(d))$
\EndWhile
\EndProcedure
\Ensure $B$
\end{algorithmic}
\end{algorithm} \end{minipage}
 \hfill
 \begin{minipage}[t]{0.39\textwidth}
\begin{algorithm}[H]
\caption{Finding high leverage rows}
\label{RowUpdates}
\begin{algorithmic}[1]
\Require Well-conditioned basis $X$ for matrix $W$, threshold parameter $\tau > 0$
\Procedure{LevScoreCheck}{$X, W, \tau$}
\State $N \leftarrow \text{Number of rows in }X$
\State $Y \leftarrow \mathbf{0}$
\For{$i = 1:N$}
    \If{$w_i(X) > \tau$}
        \State $Y_i \leftarrow W_i$
    \EndIf
\EndFor
\EndProcedure
\Ensure Nonzero rows of $Y$
\end{algorithmic}
\end{algorithm}
 \end{minipage}
 \hfill
\caption{$[X ; Y]$ denotes row-wise appending of matrices,  $U = \text{wcb}(M)$ denotes that $U$ is a well-conditioned basis for $M$.}
\end{figure*}

We prove Theorem \ref{DetLevScoreThm} by arguing the correctness of
Algorithm \ref{DetLevScoreAlg} which reads $A$ once only, row by row, and so
operates in the streaming model of computation as follows.
Let $A'$ be the submatrix of $A$
induced by the $b$ block of $\text{poly}(d)/\tau$ rows.
Upon storing $A'$, we compute $U$, a local well-conditioned basis for $A'$  and
the local leverage scores with respect to $U$, $\hat{w}_i(U)$ are calculated.
Now, the local and global leverage scores can be related by Lemma
\ref{lem: lev_scores_drop} as $w_i / \text{poly}(d) \le \hat{w}_i$ so we can
decide which rows to keep using an adjusted threshold.
Any $i$ for which the local leverage exceeds the adjusted threshold is kept in
the sample and all other rows are deleted.
The sample cannot be too large by properties of the well-conditioned basis and
leverage scores so these kept rows can be appended to the next block which is
read in before computing another well-conditioned basis and repeating in the
same fashion.
The proof of Theorem \ref{DetLevScoreThm} is deferred to Appendix \ref{sec: appendix high lev}.

\begin{thm} \label{DetLevScoreThm}
Let $\tau > 0$ be a fixed constant and let $b$ denote a bound on the
available space.
There exists a deterministic algorithm,
namely, Algorithm \ref{DetLevScoreAlg}, which computes the $\ell_p$-leverage
scores of a matrix $A \in \R^{n \times d}$ with $O(bd^2 + b d^5 \log b)$ update time,
$\text{poly}(d)$ space, and returns all rows of $A$ with $\ell_p$ leverage score
satisfying $w_i \ge \tau / \operatorname{poly}(d)$.
\end{thm}

\section{$\ell_p$-Subspace Embeddings} \label{Deterministic ell+p subspace embedding section}

Under the assumptions of the Distributed Summary Model we present an
algorithm which computes an $\ell_p$-subspace embedding.
By extension, this applies to both the distributed and streaming models of
computation as described in Section \ref{sec: models}.
Two operations are needed for this model of computation: the merge and reduce
steps.
To reduce the input at each level a summary is computed by taking a block of
input $B$ (corresponding to a leaf node or a node higher up the tree) and
computing a well-conditioned basis $B = US$.
In particular, the summary is now the matrix $S$ with $U$ and $B$ deleted.
For the merge step, successive matrices $S$ are concatenated until the space
requirement is met.
A further reduce step takes as input this concatenated matrix and the process is
repeated.
Further details, pseudocode, and proofs for this section are given in Appendix
\ref{sec: DetEllpSEthm}.

\begin{mydef}
A matrix $T$ is a relative error $(c_1, c_2)$-$\ell_p$ subspace embedding for the column space of a matrix $A \in \R^{n \times d}$ if there are constants $c_1, c_2 > 0$ so that for all $x \in \R^{d},$
%
    $c_1 \Vert Ax \Vert_p \le \Vert Tx \Vert_p \le c_2 \Vert Ax \Vert_p.$
\end{mydef}

\begin{thm} \label{DetEllpSEthm}
Let $A \in \R^{n \times d}$, $p \ne 2, \infty$ be fixed and fix a constant $\gamma \in (0,1)$.
Then there exists a one-pass deterministic algorithm which constructs a $(1/d^{O(1/ \gamma)}, 1)$ \textit{relative error} $\ell_p$-subspace embedding in with $O(n^{\gamma}d^2 + n^{\gamma} d^5 \log n^{\gamma})$ update time and $O(n^{\gamma}d)$ space in the streaming and distributed models of computation.
\end{thm}

The algorithm is used in a tree structure as follows: split input
$A \in \R^{n \times d}$ into $n^{1 - \gamma}$ blocks of size $n^{\gamma}$,
these form the leaves of the tree.
For each block, a well-conditioned basis is computed and the change of basis matrix $S \in \R^{d \times d}$ is stored and passed to the next level of the tree.
This is repeated until the concatenation of all the $S$ matrices would exceed $n^{\gamma}$.
At this point, the concatenated $S$ matrices form the parent node of the
leaves in the tree and the process is repeated upon this node: this is the
merge and reduce step of the algorithm.
At every iteration of the merge-and-reduce steps it can be shown that a
distortion of $1/d$ is introduced by using the summaries $S$.
However, this can be controlled across all of the $O(1/\gamma)$ levels in the tree to give a
deterministic relative error $\ell_p$ subspace embedding which requires only
sublinear space and little communication.
In addition, the subspace embedding can be used to achieve a deterministic relative-error
approximate regression result.
The proof relies upon analyzing the merge-and-reduce behaviour across all nodes of the tree.

$\ell_p$-\textbf{Regression Problem:} Given matrix $A \in \R^{n \times d}$ and target vector $b \in \R^n$, find $\hat{x} = \argmin_x \|Ax - b\|_p$.

\begin{thm} \label{DetEllpReg}
Let $A \in \R^{n \times d}, b \in \R^n$, fix $ p \ne 2, \infty$ and a
constant $ \gamma > 0$. The \emph{$\ell_p$-regression problem} can be
solved deterministically in the streaming and distributed models with a $(d+1)^{O(1/\gamma)}
= \text{poly}(d)$ \textit{relative error} approximation factor.
The update time is $\text{poly}(n^{\gamma}(d+1))$
and $O((1/ \gamma) n^{\gamma}(d+1))$ storage.
The query time is $\operatorname{poly}(n^{\gamma})$ for the cost of convex optimization.
\end{thm}

\section{Low-Rank Approximation} \label{low_rank_main}

$\ell_1$-\textbf{Low-Rank Approximation Problem:} Given matrix $A \in \R^{n \times d}$ output a matrix $B$ of rank $k$ s.t., for constant $k$:
\begin{equation} \label{LowRankProb}
    \Vert A - B \Vert_1 \le \text{poly}(k) \min_{A': \text{rank} k} \Vert A - A' \Vert_1.
\end{equation}

\begin{thm} \label{l1Approxthm}
Let $A \in \R^{n \times d}$ be the given data matrix
and $k$ be the (constant) target rank. 
Let $\gamma > 0$ be an arbitrary (small) constant.  Then there exists a deterministic distributed and streaming algorithm (namely Algorithm \ref{DetL1Alg} in Appendix \ref{ell_1 low rank approx}) which can output a solution to the $\ell_1$-Low Rank Approximation Problem with \textit{relative error} $\text{poly}(k)$ approximation factor, update time $\text{poly}(n,d)$, space bounded by $n^{\gamma}\text{poly}(d)$, and query time $\operatorname{poly}(n, d)$.
\end{thm}

The key technique is similar to that of the previous section by using a tree structure with merge-and-reduce operations. For input $A \in \R^{n \times d}$ and constant $\gamma > 0$ partition $A$ into $n^{1- \gamma}$ groups of rows which form the leaves of the tree.  
The tree is defined as previously with the same `merge' operation, but the `reduce' step to summarize the data exploits a derandomization (subroutine Algorithm \ref{DetKRankApprox}) of \cite{woodruff_ell_1_low_rank} to compute an approximation to the optimal $\ell_1$-low-rank approximation.
Once this is computed, $k$ of the rows in the summary are kept for later merge steps.

This process is continued with the successive $k$ rows from $n^{\gamma}$ rows being `merged' or added to the matrix until it has $n^{\gamma}$ rows.
The process is repeated across all of the groups in the level and again on the successive levels on the tree from which it can be shown that the error does not propagate too much over the tree, thus giving the desired result.

\section{Application: $\ell_{\infty}$-Regression} \label{Deterministic ell_inf regression}
\label{ell_inf_section}
Here we present a method for solving $\ell_{\infty}$-regression in a streaming
fashion.
Given input $A$ and a target vector $b$, it is possible to achieve additive
approximation error of the form $\varepsilon \| b \|_p$ for arbitrarily large $p$.
This contrasts with both Theorems \ref{DetEllpSEthm} and \ref{DetEllpReg}
which achieve a relative error $\text{poly}(d)$ approximation.
Both of these theorems require that $p$ is constant and not equal to the
$\infty$-norm.
This restriction is due to a lower bound for $\ell_{\infty}$- regression showing
 that it cannot be approximated with relative error in sublinear space.
 The key to proving Theorem \ref{ell_inf reg} below is using Theorem \ref{DetLevScoreThm}
 to find high leverage rows and arguing that these are sufficient to give the
 claimed error guarantee.

The $\ell_{\infty}$-regression problem has been previously studied in the
overdetermined case and can naturally be applied to curve-fitting under this norm.
$\ell_{\infty}$-regression can be solved by linear
programming \cite{sposito1976minimizing} and such a transformation allows
the identification of outliers in the data.
Also, if the errors are known to be distributed uniformly across an interval
then $\ell_{\infty}$-regression estimator is the maximum-likelihood parameter choice \cite{hand1978aspects}.
The same work argues that such uniform distributions on the errors often
arise as round-off errors in industrial applications whereby the error is
controlled or is small relative to the signal.
There are further applications such as using $\ell_{\infty}$-regression to
remove outliers prior to $\ell_2$ regression in order to make the problem more
robust \cite{shen2014fast}.
By applying $\ell_{\infty}$ regression on subsets of the data an approximation
to the Least Median of Squares (another robust form of regression) can be found.
We now define the problem and proceed to show that it is possible to compute an approximate solution with additive error in $\ell_p$-norm for arbitrarily large $p$.

\textbf{Approximate} $\ell_{\infty}$-\textbf{Regression problem:}
Given data $A \in \R^{n \times d}$, target vector $b \in \R^n$, and error parameter $\varepsilon > 0$, compute an additive $\varepsilon \| b \|_p$ error solution to:
\begin{equation*}
\min_{x \in \R^d} \|Ax - b \|_{\infty} = \min_{x \in \R^d} \left[ \max_{i} |(Ax)_i  - b_i | \right].
\end{equation*}


\begin{thm} \label{ell_inf reg}
Let $A \in \R^{n \times d}, b \in \R^n$ and fix constants $p \ge 1, \varepsilon > 0$ with $p \ne \infty$. There exists a one-pass deterministic streaming algorithm which solves the \emph{$\ell_{\infty}$-regression problem} up to an additive $\varepsilon \Vert b \Vert_p$ error in $d^{O(p)}/\varepsilon^{O(1)}$ space, $O(md^5 + md^2 \log m)$ update time and $T_{\text{solve}}(m,d)$ query time.
\end{thm}

Note that $T_{\text{solve}}(m,d)$ query time is the time taken to solve the linear program associated with the above problem on a reduced instance size.
Also, observe that Theorem \ref{ell_inf reg} requires $p<\infty$.  This restriction is
 necessary to forbid relative error with respect to the infinity norm.  Indeed,
 $p$ can be an arbitrarily large constant, but for $p = \infty$ we can look for
 rows above an $\varepsilon / \text{poly}(d)$ threshold in the case when $A$ is
 an all-ones column $n$-vector (so an $n \times 1$ matrix).  Then
 $\| Ax \|_{\infty} = \|x\|_{\infty}$ since $x$ is a scalar.  Also, $A$ is a
 well-conditioned basis for its own column span but the number of rows of
 leverage exceeding $\varepsilon /\text{poly}(d) = \varepsilon$ is $n$ for a
 small constant $\varepsilon$.  This intuition allows us to prove
the following theorem.  

\begin{thm} \label{thm: ell_inf lower bound}
Any algorithm which outputs an $\varepsilon \Vert b \Vert_{\infty}$ relative error solution to the $\ell_{\infty}$-regression problem requires $\min \left\{ n, 2^{\Omega(d)} \right\}$ space.
\end{thm}

\section{Experimental Evaluation} \label{sec: experiments}

To validate our approach, we evaluate the use of high $\ell_p$-leverage rows in order
to approximate $\ell_{\infty}$-regression\footnote{Code available at \url{https://github.com/c-dickens/stream-summaries-high-lev-rows}}, focusing particularly on the cases using 
$\ell_1$ and $\ell_2$ well-conditioned bases.
It is straightforward to model $\ell_{\infty}$-regression as a linear
program in the offline setting.
We use this to measure the accuracy of our algorithm.
The implementation is carried out in the single pass streaming model with a fixed space
constraint, $m$, and threshold, $\alpha^p / m$ for both 
conditioning methods to ensure the
number of rows kept in the summary did not exceed $m$. 
Recall from Remark \ref{rem: tree-stream} that the single-pass streaming implementation is equivalent to the distributed model with only one participant applying merge-and-reduce, so this experiment can also be seen as a distributed computation with the merge step being the appending of new rows and the reduce step being the thresholding in the new well-conditioned basis.

\textbf{Methods.}  We analyze two instantiations of our methods
based on how we find a well-conditioned basis and repeat
 over 5 independent trials
with random permutations of the data.
The methods are as follows:

\texttt{SPC3}: 
We use an algorithm of \citet{wcb_alg} to compute an $\ell_1$-wcb.
This method is randomized as it employs the Sparse Cauchy Transform and is only an $\ell_1$-well-conditioned basis with constant probability
We also implemented a check condition which showed that almost always, roughly $99\%$
of the time, the 
randomized construction \texttt{SPC3} would return a $(d^{2.5}, 1,1)$-well-conditioned basis.
Thus, we bypassed this check in our experiment to ensure quick update times.

\texttt{Orth}:
In addition, we also used an orthonormal basis using the QR decomposition
which is an $\ell_2$-wcb.
This method is fully deterministic and outputs a  $(\sqrt{d}, 1, 2)$-well-
conditioned basis.

\texttt{Sample}:
A sample of the data is chosen uniformly at random and the retained summary
has size exactly $m$.

\texttt{Identity}:
No conditioning is performed.
For a block $B$ of the input, the \textit{surrogate} scores
$w_i(B) = \|\mathbf{e}_i^T B \|_2^2 / \|B\|_F^2$ are used to determine
which rows to keep.
As the sum of these $w_i(B)$ is $1$, we keep all rows which have 
$w_i(B) > 2/m$.
Since no more than $m/2$ of the rows can satisfy $w_i(B) > 2/m$, the 
size of the stored subset of rows can be controlled and cannot grow 
too large.

\begin{rem}
The \texttt{Identity} method keeps only the rows with high
norm which contrasts our conditioning approach: if most of the 
mass of the block is concentrated on a few rows then
these will appear heavy locally despite the possibility that 
they may correspond to previously seen or unimportant directions.
In particular, if these heavy rows significantly outweigh the weight
of some sparse directions in the data it is likely that the sparse 
directions will not be found at all.
For instance, consider data $X \in \R^{n \times d}$ which is then
augmented by appending the identity (and zeros) so that 
these are the only vectors in the new directions.
That is, set $X' = [X, \mathbf{0}_{n \times k} ; \mathbf{0}_{k \times d}, I_{k \times k}]$
and then permute the rows of $X'$. 
The appended sparse vectors from $I_{k \times k}$ will have leverage of 
$1$ so will be detected by the well-conditioned basis methods. 
However there is no guarantee that the \texttt{Identity} method will identify
these directions if the entries
in $X$ significantly outweigh those in $I_{k \times k}$.
In addition, there is also no guarantee that using uniform sampling will 
identify these points, particularly when $k$ is small compared to $n$ and $d$.
So while choosing to do no conditioning seems attractive, this example
shows that doing so may not give any meaningful guarantees and 
hence we prefer the approach in Section \ref{HighLevScores}.
We compare only to these baselines as we are not aware of any other competing
methods in the small memory regime for the $\ell_{\infty}$-regression problem.
\end{rem}

\textbf{Datasets.} We tested the methods on a subset of the \emph{US Census Data} containing 5 million rows and 11 columns\footnote{\url{http://www.census.gov/census2000/PUMS5.html}} and 
 \emph{YearPredictionMSD}\footnote{\url{https://archive.ics.uci.edu/ml/datasets/yearpredictionmsd}} which has roughly 500,000 rows and 90 columns
 (although we focus on a fixed 50,000 row sample so that the LP for regression is tractable: see Figure~\ref{fig: years_regression_time_vs_block_size.pdf} in the Supplementary Material, Appendix \ref{sec: experimental results appendices}).
For the census dataset, space constraints between 50,000 and 500,000
rows were tested and for the YearPredictionsMSD data space budgets were
tested between 2,500 and 25,000.
The general behavior is roughly the same for both datasets so
for brevity we
primarily show the results for \emph{US Census Data}, and defer
corresponding plots for \emph{YearPredictionsMSD} to Appendix~\ref{sec: experimental results appendices}.

\begin{figure*}[t]
\centering
\subfloat[U.S. Census]{\includegraphics[width=0.33\textwidth]{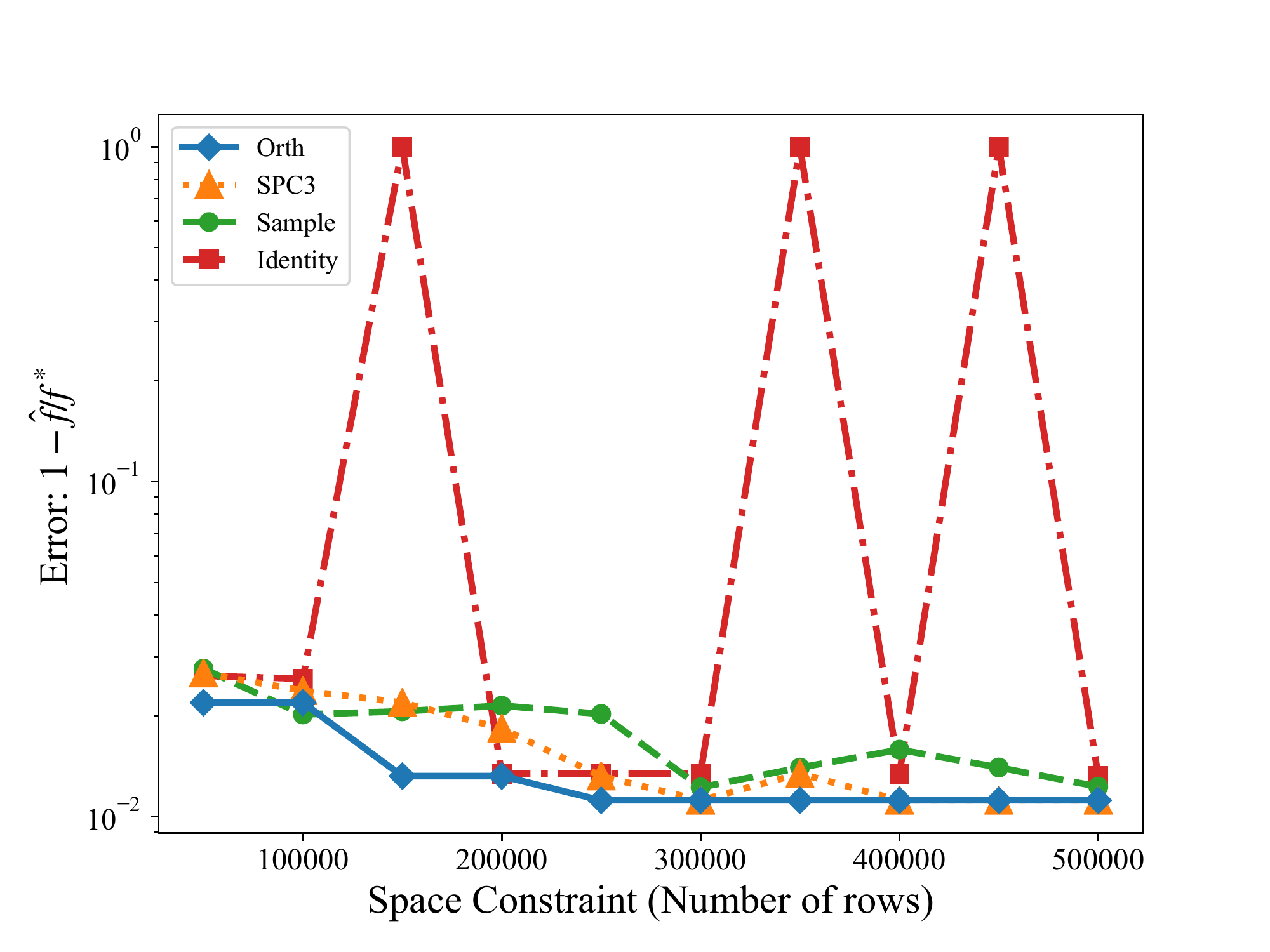}
\label{fig: census_error_vs_block_size.pdf}}%
\subfloat[YearPredictionMSD]{
\includegraphics[width=0.32\textwidth]{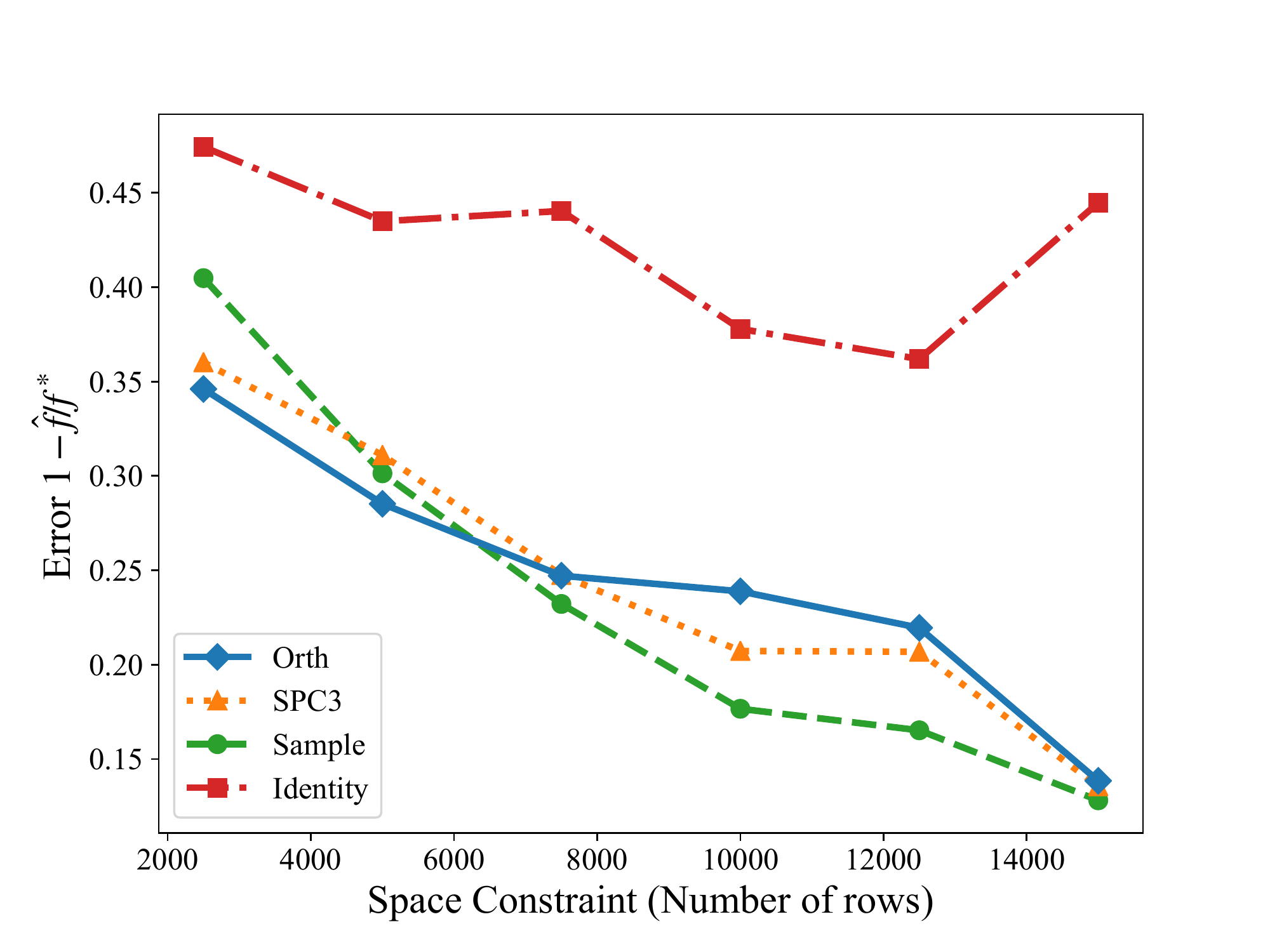}
\label{fig: years_error_vs_block_size.pdf}}%
\subfloat[Max summary size vs space constraint]{
\includegraphics[width=0.32\textwidth]{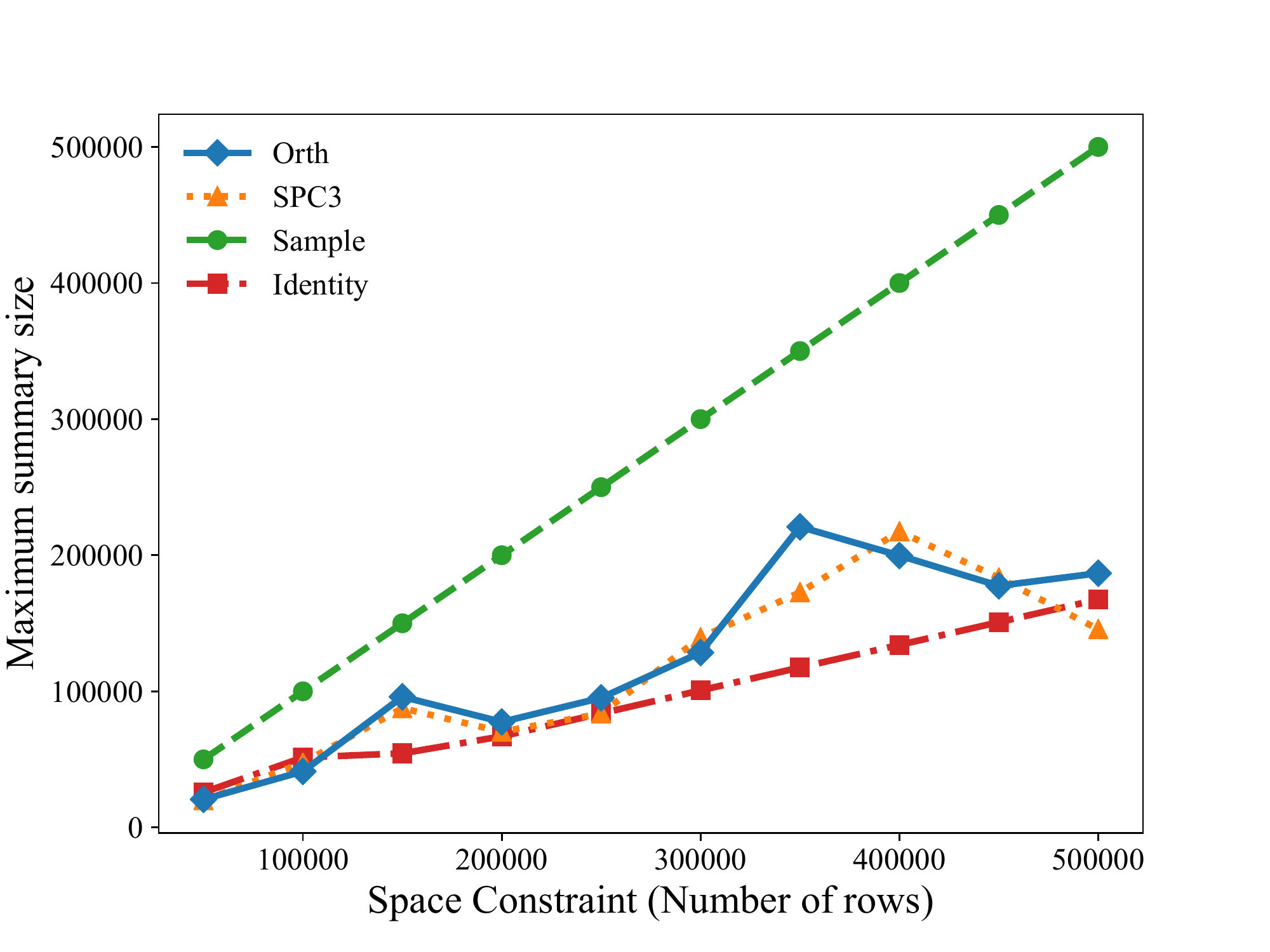}
\label{fig: census_space_vs_block_size.pdf}}

\caption{Error vs Space Constraint in (a) and (b) and Maximum Summary Size vs Space Constraint (c).  Total input size is $5000000 \times 11$.}
\end{figure*}

\textbf{Results on approximation error compared to storage}
Let $f^*$ denote the minimal value of the full regression obtained by $x^*$ and let $x'$ be the output of the reduced problem.
The approximate solution to the full problem is then $\hat{f} = \|Ax'-b\|_{\infty}$
and approximation error is measured as $\hat{f}/f^* - 1$
(note that $\hat{f} \ge f^*$).
An error closer to $0$ demonstrates that $\hat{f}$ is roughly the same as $f^*$ so the optimal value is well-approximated.
Figures \ref{fig: census_error_vs_block_size.pdf} and 
\ref{fig: years_error_vs_block_size.pdf}
show that on both datasets the \texttt{Identity} method 
consistently performs poorly while \texttt{Sample} achieves 
comparable accuracy to the conditioning methods.  
Despite the simplicity of uniform sampling to keep a summary, the succeeding sections discuss the increased time and space costs of using such a sample and show that doing so is not favourable.   
Thus, neither of the baseline methods output a summary which can be used to approximate the 
regression problem both \textit{accurately} and \textit{quickly},
hence justifying our use of leverage scores. 
 Our conditioning methods perform particularly well in the \emph{US Census Data} data (Figure \ref{fig: census_error_vs_block_size.pdf})
 with \texttt{Orth} appearing to give the most accurate summary and 
 \texttt{SPC3} performing comparably well but with
 slightly more fluctuation: similar behaviour is observed in the \textit{YearPredictionMSD} (Figure \ref{fig: years_error_vs_block_size.pdf}) data too.
 The conditioning methods are also seen to be robust to the storage constraint, give accurate performance across both datasets using significantly
 less storage than sampling, and give a better estimate in general than doing no conditioning.

\textbf{Results on Space Complexity.}
Recall that the space constraint is $m$ rows and throughout the stream, after a local computation, the merge step concatenates more rows to the existing summary until
the bound $m$ is met, prior to computing the next reduction.
During the initialization of the block $A'$ by Algorithm \ref{DetLevScoreAlg},
the number of stored rows is exactly $m$. However, we measure the
maximum number of rows kept in a summary after every reduction step
to understand how large the returned summary can grow.
As seen in Figure \ref{fig: census_space_vs_block_size.pdf},
\texttt{Identity} keeps the smallest summary but there is no reason
to expect it has kept the most important rows.
In contrast, if $m$ is the bound on the summary size, then uniform sampling
always returns a summary of size exactly $m$.
However, we see that this is not optimal as both conditioning methods can 
return a set of rows which are pruned at every iteration to roughly 
half the size and contains only the most important rows in that block.
Both conditioning methods exhibit similar behavior and are bounded 
between both \texttt{Sample} and \texttt{Identity} methods.
Therefore, both of the conditioning methods
respect the theoretical bound and, crucially, return a summary which is 
sublinear in the space constraint and hence a significantly smaller
fraction of the input size.

\eat{
In contrast 

As seen in Figure \ref{fig: census_space_vs_block_size.pdf}, the sampling procedure always keeps a summary of size $m$ (where $m$ is the number of rows to randomly choose), however, despite having a space budget of $m$, the conditioning methods always prune the stored set of rows down to roughly half the size.
Therefore, only the most important rows are kept at each step so much of the information which is not informative can be discarded and a summary can be returned which is a fraction of the space limit as well as the overall size of the problem.

While using no conditioning seems attractive from this perspective we will later discuss why this is not desirable.}

\begin{figure*}[t] 
\subfloat[Time to compute local basis \label{fig: census_basis_time_vs_block_size.pdf}]{%
\includegraphics[width=0.33\textwidth]{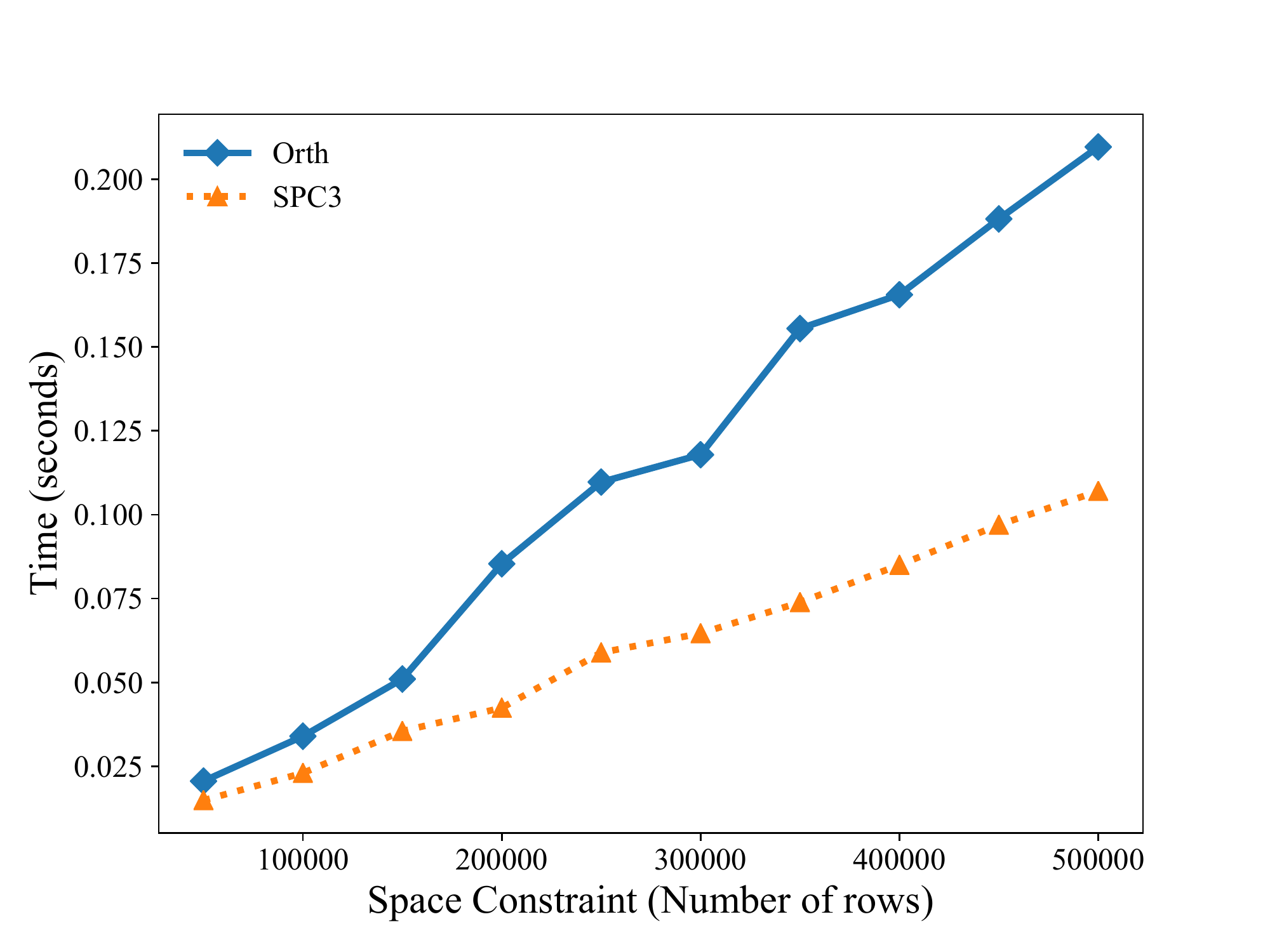}}%
\subfloat[Solution time for $\ell_{\infty}$-regression \label{fig: census_regression_time_vs_block_size.pdf}]{%
\includegraphics[width=0.33\textwidth]{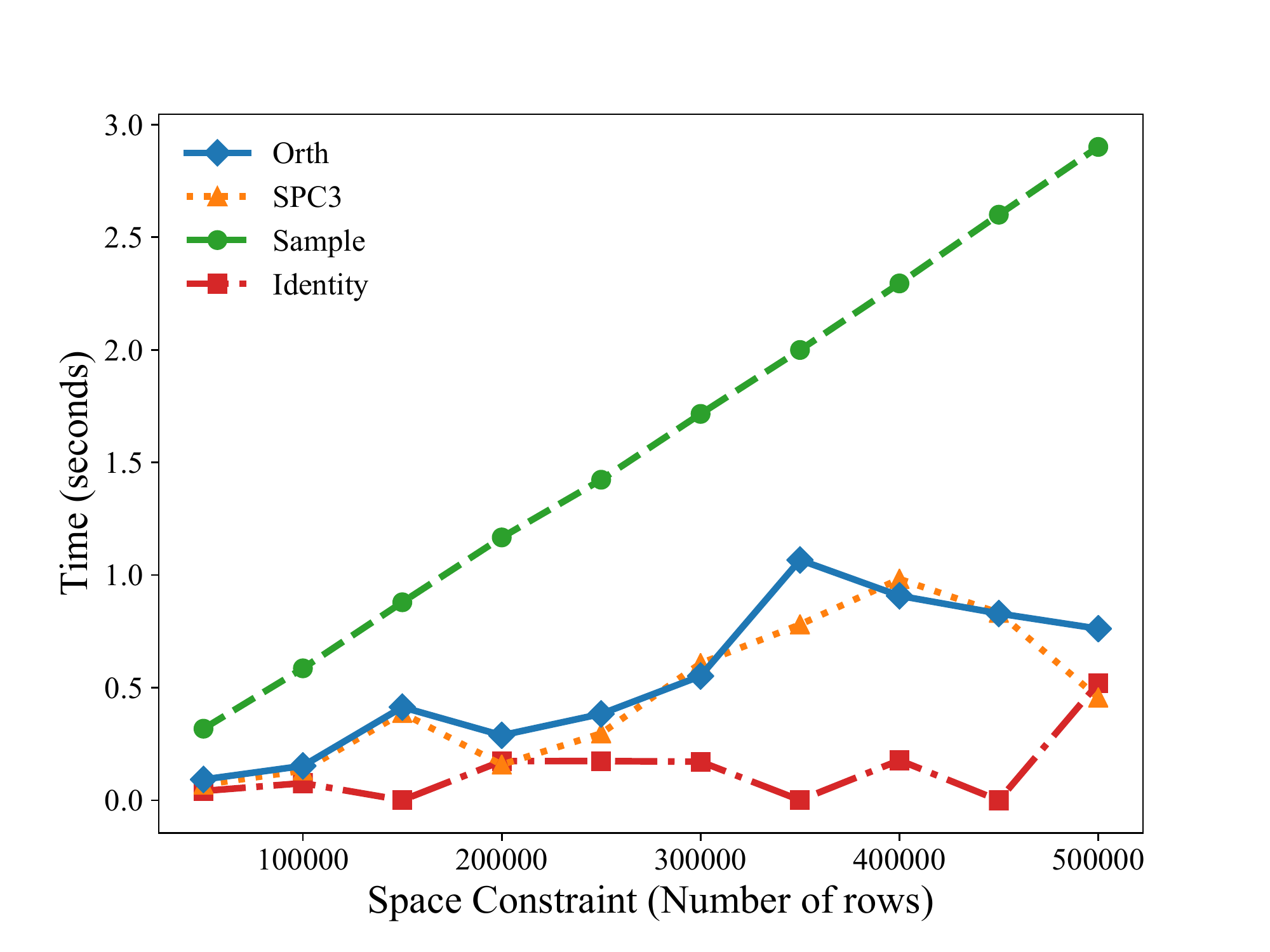}}%
\subfloat[Total time \label{fig: census_total_time_vs_block_size.pdf}]{%
\includegraphics[width=0.33\textwidth]{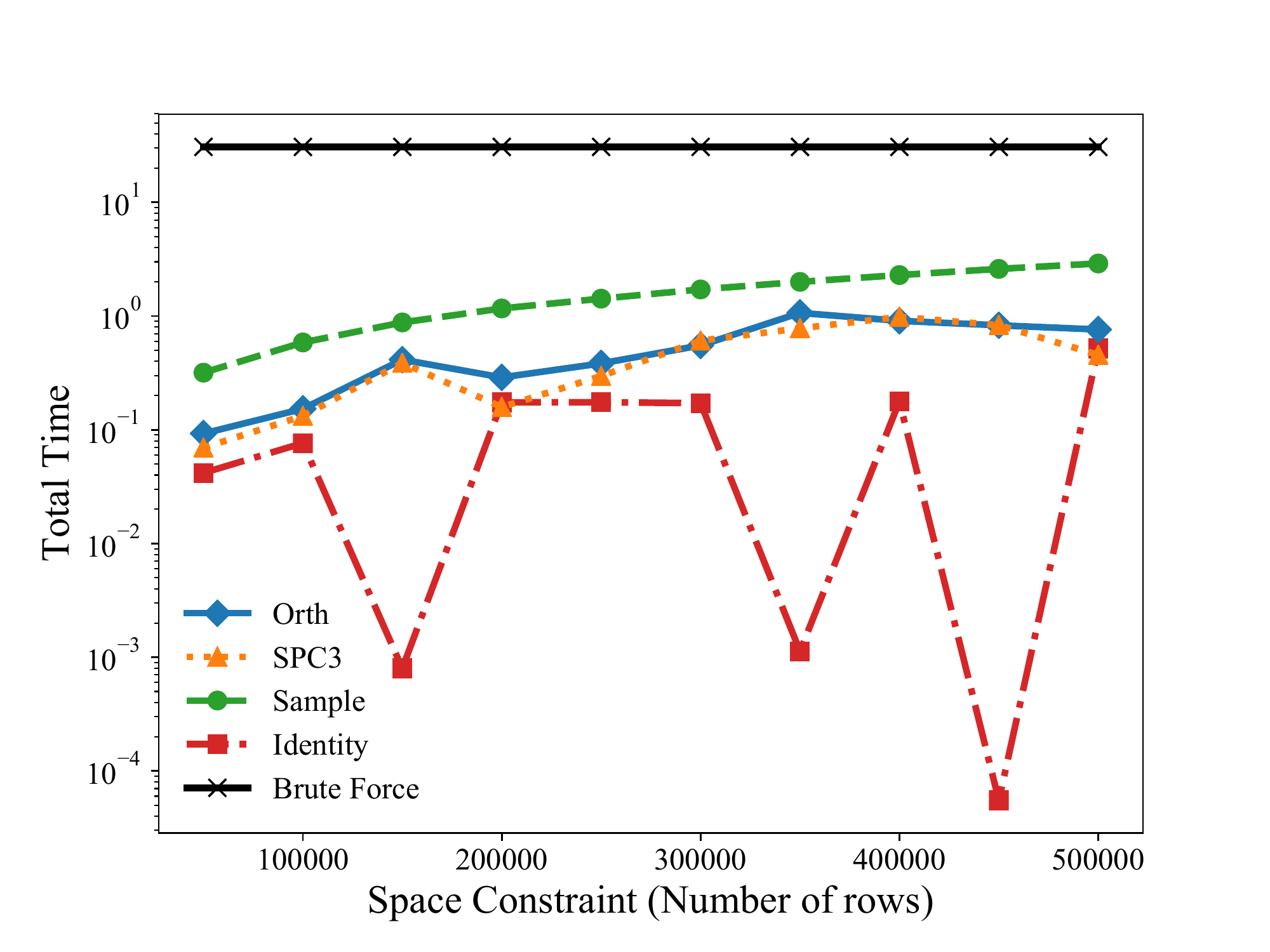}}
\caption{Computation Times compared to summary size\label{fig: summary_times}}
\end{figure*}

\textbf{Results on Time Complexity.}
There are three time costs measured.  
The first is the update time taken to compute the local well-conditioned basis which
is theoretically $O(md^2 + md^5 \log m)$ by Theorem \ref{WCBthm}.
However, the two bases that we test are an orthonormal basis, computable
 in time $O(md^2)$ and the \texttt{SPC3} transform which takes time
  $O(\text{nnz}(B) \log m)$ for a block $B$ with $m$ rows and
  $\text{nnz}(B)$ non-zero entries.
  Figure \ref{fig: census_basis_time_vs_block_size.pdf} demonstrates that
  \texttt{SPC3} is faster than \texttt{Orth} on this data in practice but
  this small absolute difference becomes negligible over the entirety of the stream
  as seen in Figure \ref{fig: census_total_time_vs_block_size.pdf}.
  The query time in Figure \ref{fig: census_regression_time_vs_block_size.pdf}
  is roughly proportional to the summary size in all instances but here
  the conditioning methods perform noticeably better due to the smaller summary size that is
  returned as discussed in the previous section.
  However, as seen in Figure \ref{fig: years_regression_time_vs_block_size.pdf}, 
  (Supplementary Material, Appendix \ref{sec: experimental results appendices} ) this disparity becomes 
  hugely significant on higher dimensionality data due to the increased size summary retained 
  by sampling, further justifying our approach of pruning rows at every stage.
  While \texttt{Identity} appears to have fast query time, this is due to the summary being smaller.
  Although it may seem that for smaller summaries more local bases need to be
  computed and this time could prohibitively increase over the stream,
   Figure \ref{fig: census_total_time_vs_block_size.pdf} demonstrates 
   that even using small blocks does not cause the \textit{overall time} (to process the 
  stream and produce an approximate query) to increase too much.
  Hence, an approximation can be obtained which is highly accurate, and 
  in total time faster than the brute force solver.

\textbf{Experimental Summary.}
While it might seem attractive not to perform any conditioning on 
the matrix and just pick heavy rows, 
our experiments show that this strategy is not effective in practice,
and delivers poor accuracy. 
Although a simple sample of randomly chosen rows can be easily
maintained, this appears less useful due to the increased
time costs associated with larger summaries when conditioning methods
output a similar estimate in less time over the entire stream.
As the $\ell_{\infty}$-regression problems depend only on a few rows
of the data there are cases when uniform sampling can perform well:
if many of the critical rows look similar then there is a chance that uniform
sampling will select some examples.
In this case, the leverage of the important 
direction is divided across the repetitions,
and so it is harder to ensure that desired direction is identified. 
Despite this potential drawback we have shown that  both \texttt{Orth} 
and \texttt{SPC3} can be used to find accurate summaries which 
perform robustly across each of the measures we have tested.
It appears that \texttt{SPC3} performs comparably to \texttt{Orth}; both are
relatively quick to compute and admit accurate summaries in similar space.
In particular, both conditioning methods return summaries which are a fraction
of the space budget and hence highly sublinear in the input size, 
which give accurate approximations and are robust to the concatenation of
new rows. 
All of these factors make the conditioning method fast in practice to both find the important rows in the data and then compute the reduced regression problem with high accuracy.

Due to the problems in constructing summaries which can be used to solve regression
quickly and accurately when using random sampling or no transformation, our methods are shown to be efficient and accurate alternatives.
Our approach is vindicated both theoretically and practically: this is most clear in the U.S. Census dataset where small error can be achieved using a summary roughly $2\%$ the size of the data.
This also results in an overall speedup as solving the optimization on the reduced set is much faster than solving on the full problem. 
Such significant savings show that this general approach can be useful in large-scale applications.

\clearpage

\section*{Acknowledgements}
The work of G. Cormode and C. Dickens is supported by
European Research Council grant ERC-2014-CoG 647557 and
The Alan Turing Institute under the EPSRC grant EP/N510129/1. 
D. Woodruff would like to acknowledge the support by the National 
Science Foundation under Grant No. CCF-1815840.

\bibliography{deterministic}
\bibliographystyle{icml2018}

\clearpage

\appendix

\section*{Supplementary Material for Leveraging Well-Conditioned Bases: Streaming and \\
        Distributed Summaries in Minkowski $p$-Norms}

\section{Proofs for Section \ref{HighLevScores}} \label{sec: appendix high lev}

\begin{Lemma} \label{lem: lev_scores_drop}
Denote the $i$th global leverage score of $A$ by $w_i$ and its associated
local leverage score in a block of input $A$ be denoted $\hat{w}_k$.
Then $w_i / \text{poly}(d) \le \hat{w}_k$.  In particular, $w_i /
(d \alpha^{p} \beta) \le \hat{w}_k$.

\end{Lemma}

\begin{myproof}{Lemma}{\ref{lem: lev_scores_drop}}
{Let $U = AR$.  Recall that $w_i = \Vert \mathbf{e}_i^T AR \Vert_p^p$.
Then for some coordinate $j$ we must have $|\mathbf{e}_i^T AR \mathbf{e}_j |^p
\ge w_i/d$.  Taking $\mathbf{x} =  \mathbf{e}_j$ we see that

\begin{equation} \label{bound1}
    |(AR \mathbf{x})_i |^p \ge \frac{w_i}{d}.
\end{equation}
However, Fact \ref{LevBound1} implies :

\begin{equation} \label{bound2}
    \Vert AR \mathbf{x} \Vert_p^p \le \Vert AR \Vert_p^p \le  \alpha^p \le \text{poly}(d).
\end{equation}

Hence, there exists a $\mathbf{y} \in \text{col}(A)$ with $\mathbf{y} = AR\mathbf{x}$ such that $|\mathbf{y}_i|^p \ge w_i/d$ from Equation (\ref{bound1}).  Also, $\Vert \mathbf{y} \Vert_p^p \le \alpha^p$ from Equation (\ref{bound2}).  Thus,
%
    \begin{align}
        \frac{|\mathbf{y}_i|^p}{\Vert \mathbf{y} \Vert_p^p} &\ge \frac{|\mathbf{y}_i|^p}{\alpha^p} \\
        &\ge \frac{w_i}{d \alpha^p} \\
        &\ge \frac{w_i}{\text{poly}(d)}.
    \end{align}

From this we see;
\begin{equation} \label{eq: y_indexBound}
    w_i \le \frac{|\mathbf{y}_i|^p d \alpha^p}{\Vert \mathbf{y} \Vert_p^p}.
\end{equation}

Now, let $B$ be a block of rows from $A$.  We manipulate $B$ by considering it either as an individual matrix or as a coordinate subspace of $A$; i.e all rows are zero except for those contained in $B$ which will be denoted by $\hat{A}$.  Define $\hat{\mathbf{y}} = \hat{A}R\mathbf{x}$.  Then $\hat{\mathbf{y}}_{j'} = \mathbf{y}_{j'}$ when $j'$ is a row from $B$ and $\hat{\mathbf{y}}_{j'} = 0$ otherwise.  Thus, $\|\hat{\mathbf{y}} \|_p^p \le \| \mathbf{y} \|_p^p$ and:

\begin{equation} \label{eq: intermediate}
    w_i \le \frac{|\mathbf{y}_i|^p d \alpha^p}{\Vert \hat{\mathbf{y}} \Vert_p^p}.
\end{equation}

For rows $i$ which are also found in $B$ (indexed as $k$) we see that $|\hat{\mathbf{y}}_{k}|^p = |\mathbf{y}_{i}|^p$.  So, \emph{for such indices}, using Equations (\ref{eq: y_indexBound}) and (\ref{eq: intermediate}):

\begin{equation} \label{eq: intermediate2}
     w_i \le \frac{|\hat{\mathbf{y}}_k|^p d \alpha^p}{\Vert \hat{\mathbf{y}} \Vert_p^p}.
\end{equation}

Since $\hat{\mathbf{y}}$ is the restriction of $\mathbf{y}$ to coordinates of $B$ we can write $\hat{\mathbf{y}} = B \hat{R} \hat{x}$ where $B \hat{R}$ is well-conditioned.  Let $\hat{w}_k = \| \mathbf{e}_k^T B \hat{R} \|_p^p$ be the $k$th local leverage score in $B$.  By applying the same argument as in Fact \ref{LevBound2} it can be shown that $|\hat{\mathbf{y}}_{k}|^p/\|\hat{\mathbf{y}}\|_p^p \le \text{poly}(d) \hat{w}_k$.  Indeed,
%
    \begin{align}
        |\hat{\mathbf{y}}_k|^p &= | ( \hat{A} \hat{R} \hat{x})_k|^p \\
                      &\le \| e_k^T \hat{A} \hat{R} \|_p^p \|\hat{x}\|_q^p  \text{\quad by H\"{o}lder's inequality} \\
                      &\le \hat{w}_k \beta \| B \hat{R} \hat{x} \|_p^p \\
                      &\le  \beta \hat{w}_k \|\hat{\mathbf{y}}\|_p^p.
    \end{align}

The second inequality uses condition 2 from Theorem \ref{WCB} and the fact that $B \hat{R}$ is a well-conditioned basis.
Then using Equation \ref{eq: intermediate2}, the following then proves the latter claim of the lemma:
\begin{equation*}
\frac{w_i}{d \alpha^p} \le \frac{|\hat{\mathbf{y}}_k|^p}{\Vert \hat{\mathbf{y}} \Vert_p^p} \le  \beta \hat{w}_k.
\end{equation*}

Finally, Theorem \ref{WCBthm} states that $\beta$ is at most $\text{poly}(d)$ which proves the result.}
\label{sec: lev_scores_drop}
\end{myproof}

\begin{Lemma} \label{lem: space_local_vs_global}
All global leverage scores above a threshold can be found by computing local
leverage scores and increasing the space complexity by a $\text{poly}(d)$ factor.


\end{Lemma}

\begin{myproof}{Lemma}{\ref{lem: space_local_vs_global}}{

First we determine the space necessary to find all leverage scores exceeding $\delta$.
Let $I = \{i: w_i > \delta\}$.
Then $\alpha^p \ge \sum_{i=1}^n w_i \ge \sum_{i \in I} w_i \ge \delta |I|$ by arguing as in Fact \ref{LevBound1}
Hence, the space necessary is $|I| \le \alpha^p / \delta$.
Now focus on finding these rows in the streaming fashion.
By Lemma \ref{lem: lev_scores_drop} we see that for rows $k$ in the block which
is stored from the stream we have the property that $w_i / d \alpha^p \beta \le \hat{w}_k$.
Hence, any $w_i > \delta$ results in $\hat{w}_k >\delta / d \alpha^p \beta$ for the local thresholding.
So to keep all such $w_i > \delta$, we must store all $\hat{w}_k > \delta /
d \alpha^{p}\beta = \hat{\delta}$.
Arguing similarly as in Fact \ref{LevBound1} again define $\hat{I} = \{k : \hat{w}_k > \hat{\delta} \}$ so that: $\alpha^p \ge \sum_{k} \hat{w}_k \ge \sum_{k \in \hat{I}} \hat{w}_k \ge \hat{\delta} |\hat{I}|$.
Hence $|\hat{I}| \le \alpha^p / \hat{\delta} = d \alpha^{2p} \beta  / \delta$.
That is, $|\hat{I}| \le d \beta \alpha^p \cdot |I|$ which proves the claim as Theorem \ref{WCBthm} states that all of the parameters are $\text{poly}(d)$.
}
\end{myproof}

\begin{myproof}{Theorem}{\ref{DetLevScoreThm}}{
We claim that the output of Algorithm \ref{DetLevScoreAlg} is a matrix $B$ which
 contains rows of high leverage in $A$.
The algorithm initially reads in $b$ rows and inserts these to matrix $A'$.
A well-conditioned basis $U$ for $A'$ is then computed using Theorem
\ref{WCBthm} and incurs the associated $O(bd^2 + b d^5 \log b)$ time.
The matrix $U$ and $A'$ are passed to Algorithm \ref{RowUpdates} whereby if a
row $i$ in $U$ has local leverage exceeding $\tau$ then row $i$ of $A'$ is kept.
There are at most $\text{poly}(d)/\tau$ of these rows as seen in Lemma
\ref{lem: space_local_vs_global} and the space required is $\text{poly}(d)$ by
the same lemma.
So on the first call to Algorithm~\ref{RowUpdates} a matrix is
returned with rows whose $\ell_p$ local leverage satisfies $w_i / \text{poly}(d)
 \le \hat{w}_i$ (where $w_i$ is the global leverage score and $\hat{w}_i$
 is the associated local leverage score) and only those exceeding $\tau/
 \operatorname{poly}(d)$ are kept.


The algorithm proceeds by repeating this process on a new set of rows from $A$
and an improved matrix $B$ which contains high leverage rows from $A$ already
found.  Proceeding inductively, we see that when Algorithm \ref{RowUpdates} is
called with matrix $[A'; B]$ then a well-conditioned basis $U$ is computed.
Again $[A'; B]_i$ is kept if and only if the local leverage score from $U$,
$w_i(U) > \tau$.  By Lemma \ref{lem: space_local_vs_global} this requires
$\text{poly}(d)$ space and the local leverage score is at least a
$1/ \text{poly}(d)$ factor as large as the global leverage score by Lemma
\ref{lem: lev_scores_drop}.
Repeating  over all blocks $B$ in $A$, only the rows of high leverage are kept.
Any row of leverage smaller than $\tau / \text{poly}(d)$ is ignored so this is
the additive error incurred.}
\end{myproof}

\section{Proofs for Section \ref{Deterministic ell+p subspace embedding section}} \label{sec: DetEllpSEthm}

\textbf{Algorithm and Discussion}

The pseudocode for the first level of the tree structure of the deterministic $\ell_p$ subspace embedding described in Section \ref{Deterministic ell+p subspace embedding section} is given in Algorithm \ref{alg: deterministic-subspace-embedding}.
We use the following notation: $m$ is a counter to index the block of input currently held, denoted $A_{[m]}$, and ranges from $1$ to $n^{1-\gamma}$ for the first level of the tree.
Similarly, $t$ indexes the current summary, $P^{(t)}$ which are all initialized to be an empty matrix.
Again we use the notation $[X ; Y]$ to denote the row-wise concatenation of two matrices $X$ and $Y$ with equal column dimension.

Note that Algorithm \ref{alg: deterministic-subspace-embedding} can be easily distributed as any block of sublinear size can be given to a compute node and then a small-space summary of that block is returned to continue the computation.
In addition, the algorithm can be performed using sublinear space in the streaming model because at any one time a summary $T$ of the input can be computed which is of size $d \times d$.
Upon reading $A_{[1]}$, a small space summary $P^{(1)}$ is computed and stored with the algorithm proceeding to read in $A_{[2]}$.
Similarly, the summary $P^{(2)}$ is computed and if $[P^{(1)} ; P^{(1)}]$ does not exceed the storage bound, then the two summaries are merged and this process is repeated until at some point the storage bound is met.
Once the summary is large enough that it meets the storage bound, it is then \textit{reduced} by performing the well-conditioned basis reduction (line (\ref{alg: wcb})) and the reduced summary is stored with the algorithm continuing to read and summarize input until a corresponding block in the tree is obtained (or the blocks can be combined to terminate the algorithm).

\begin{algorithm}[t]
\caption{Deterministic $\ell_p$ subspace embedding}  \label{alg: deterministic-subspace-embedding}
\begin{algorithmic}[1]
\Procedure{$\ell_p$-SubspaceEmbedding}{$A,p,\gamma< 1)$}
\State Counters $m,t \leftarrow 1$ \label{alg: counter-step}
\State Summaries $P^{(t)} \leftarrow \text{EMPTY}$ for all $t$.
\For{$m = 1:n^{1-\gamma}$}
    \State $A_{[m]} = US$ \# $U$ an $\ell_p$ wcb for $A$ \label{alg: wcb}
   \If{$\text{num. rows}(P^{(t)}) + d \le n^{\gamma}$}
       \State{$P^{(t)} \leftarrow [P^{(t)} ; S]$}
   \Else
       \State{$P^{(t+1)} \leftarrow S$}
       \State{$t \leftarrow t + 1$} \label{alg: final-loop-step}
   \EndIf 
\EndFor
\State{\textit{Merge} all $P^{(t)}$: $T = [P^{(1)} ; \ldots ; P^{(\cdot)}]$}
\State{\textit{Reduce} $T$ by splitting into blocks of $n^{\gamma}$ and repeating lines (\ref{alg: counter-step}) - (\ref{alg: final-loop-step}) with $T$ in place of $A$.}
\State \Return $T$
\EndProcedure
\end{algorithmic}
\end{algorithm}

\begin{myproof}{Theorem}{\ref{DetEllpSEthm}}{
Let $A \in \R^{n \times d}$ and $B \in \R^{n^{\gamma} \times d}$.  We compute an $\ell_p$ well-conditioned basis for $B$ in time $\text{poly}(n^{\gamma}d)$ by Theorem \ref{WCBthm}; so let $B = US$ for $U \in \R^{n^{\gamma} \times d}$ and $S \in \R^{d \times d}$ a change of basis matrix.

From \cite{Mahoney}, $U$ satisfies $\Vert x \Vert_p \le \Vert Ux \Vert_p \le d \Vert x \Vert_p$.  This is because $\Vert x \Vert_2 \le \Vert Ux \Vert_p \le \sqrt{d} \Vert x \Vert_2$.  There are then two cases: if $p < 2$ then

\begin{equation*}
        \frac{\Vert x \Vert_p}{\sqrt{d}} \le \Vert x \Vert_2 \le \Vert Ux \Vert_p \le \sqrt{d}\Vert x \Vert_2 \le \sqrt{d} \Vert x \Vert_p
\end{equation*}

so that $\Vert x \Vert_p \le \Vert  Ux \Vert_p \le d \Vert x \Vert_p$ by rescaling by $\sqrt{d}$.  The third inequality is from \cite{Mahoney}.  Similarly, if $p > 2$ then

\begin{equation*}
        \Vert x \Vert_p \le \Vert x \Vert_2 \le \Vert Ux \Vert_p \le \sqrt{d}\Vert x \Vert_2 \le d \Vert x \Vert_p
\end{equation*}

from which $\Vert x \Vert_p \le \Vert  Ux \Vert_p \le d \Vert x \Vert_p$.  Next, the algorithm ignores $U$ and retains only $S$ after computing the well-conditioned basis.  Using the above two bounds we readily see that $\Vert Sx \Vert_p \le \Vert USx \Vert_p = \Vert Bx \Vert_p$.  Also, $\Vert Sx \Vert_p \ge \Vert USx \Vert_p / d = \Vert Bx \Vert_p / d.$  Now we have obtained a matrix $S$ which satisfies:

\begin{equation}
\frac{\Vert Bx \Vert_p}{d} \le \Vert S x \Vert_p \le \Vert Bx \Vert_p.
\end{equation}
So $\Vert Sx \Vert_p$ agrees with $\Vert Bx \Vert_p$ up to a distortion factor of $d$.

Algorithm \ref{alg: deterministic-subspace-embedding} applies the \emph{merge and reduce} framework.  The matrix $A$ is seen a row at a time and $n^\gamma$ rows are stored which are used to construct a tree.  So at every level a subspace embedding with distortion $d$ is constructed.  This error propagates through each of the $O(1/\gamma)$ levels in the tree so the overall distortion to construct the subspace embedding for $A$ is $d^{O(1/\gamma)}$.  The space bound is similar; we need $n^{\gamma}d$ storage per group so require $O(1/\gamma) n^{\gamma}d$ overall. }

\end{myproof}

\begin{myproof}{Theorem}{\ref{DetEllpReg}}
{The task is to minimise $\Vert Ax - b \Vert_p$.  Let $Z = [A,b] \in \R^{n \times (d+1)}$ and compute a subspace embedding $S$ for $Z$ using Theorem \ref{DetEllpSEthm}.  Note that $R$ has $O(1/\gamma)n^{\gamma}(d+1)$ rows.  Let $\Delta = (d+1)^{O(1/\gamma)}$, then for all $y \in \R^{d+1}$ we have:

\begin{equation} \label{ell_p_regression_subspace}
    \frac{\Vert Zy \Vert_p}{\Delta} \le \Vert Sy \Vert \le \Vert Zy \Vert.
\end{equation}

Since this condition holds for all $y \in \R^{d+1}$ it must hold, in particular, for vectors $y' = (x, -1)^T$ where $x \in \R^d$ is arbitrary.  However, observe that:

\begin{equation} \label{ell_p_regression_vector}
    \Vert Zy' \Vert_p =  \left \Vert [A,b] \begin{bmatrix}
           x \\
           -1
         \end{bmatrix} \right \Vert_p = \Vert Ax - b \Vert_p.
\end{equation}

Denote the first $d$ columns of $S$ by $S_{1:d}$ and the last column by $S_{d+1}$.  Then

\begin{equation} \label{ell_p_regression_instance}
    \Vert Sy' \Vert_p =  \left \Vert [S_{1:d},S_{d+1}] \begin{bmatrix}
           x \\
           -1
         \end{bmatrix} \right \Vert_p = \Vert S_{1:d}x - S_{d+1} \Vert_p.
\end{equation}

Now we have transformed the subspace embedding relationship into an instance of regression.  In particular, $S_{1:d}$ has only $O(1/\gamma)n^{\gamma}d$ rows so is a smaller instance than the original problem.  We now focus on the task of finding $\min_{x\in \R^{d}} \Vert S_{1:d}x - S_{d+1} \Vert_p$.  By using Equation (\ref{ell_p_regression_subspace}) with $y'$ and utilising Equations $(\ref{ell_p_regression_vector}), (\ref{ell_p_regression_instance})$ we have:

\begin{equation} \label{ell_p_regression_problem}
    \frac{\Vert Ax - b \Vert_p}{\Delta} \le \Vert S_{1:d}x - S_{d+1} \Vert_p \le \Vert Ax - b \Vert_p.
\end{equation}

Convex optimisation can now be used to find $\min_{x \in \R^d}\Vert S_{1:d}x - S_{d+1} \Vert_p$. Let $\hat{x} = \argmin_{x\in \R^d} \Vert S_{1:d}x - S_{d+1} \Vert_p$ which is output from the optimisation and let $x^* =  \argmin_{x\in \R^d} \Vert Ax - b \Vert_p$ be the optimal solution we would like to estimate.  By optimality of $\hat{x}$ we have:

\begin{equation} \label{ell_p_regression_optimality}
    \Vert S_{1:d}\hat{x} - S_{d+1} \Vert_p \le \Vert S_{1:d}x^* - S_{d+1} \Vert_p.
\end{equation}

However, combining Equation (\ref{ell_p_regression_optimality}) with Equation (\ref{ell_p_regression_problem}) we see that:

    \begin{align}
        \frac{\Vert A\hat{x} - b \Vert_p}{\Delta} &\le \Vert S_{1:d}\hat{x} - S_{d+1} \Vert_p \\
        &\le \Vert S_{1:d}x^* - S_{d+1} \Vert_p \\
        &\le \Vert Ax^* - b \Vert_p
    \end{align}

Therefore, $\Vert A\hat{x} - b \Vert_p \le \Delta \Vert Ax^* - b
\Vert_p$ and $\Delta = \text{poly}(d+1)$ so the $\ell_p$-regression
problem has been solved up to a polynomial $d+1$ approximation factor.
The overall time complexity is the time taken to compute the subspace
embedding, which is $\text{poly}(nd)$ by Theorem \ref{DetEllpSEthm},
and the time for the convex optimisation.  However, the optimisation
costs $\text{poly}(O(1/ \gamma) n^{\gamma})$~\cite{Woodruff:Zhang:13}
which is subsumed by the dominant time cost for computing the embedding.  Finally, the space cost is immediate from computing the subspace embedding in Theorem \ref{DetEllpSEthm}.} \label{sec: proof_ DetEllpReg}

\end{myproof}

\section{Proofs for Section \ref{low_rank_main} } \label{ell_1 low rank approx}
To prove correctness of Algorithm \ref{DetL1Alg} for Theorem \ref{l1Approxthm} we will need to invoke the following algorithm at each level of the tree.  This is a derandomized version of an algorithm which returns a low rank approximation to an input matrix.  The derandomization follows from generating and testing all possible combinations of the necessary matrices.

\begin{algorithm}[t]
\caption{Deterministic $\ell_1$ low rank approximation (derandomized version of algorithm from \cite{woodruff_ell_1_low_rank})}  \label{DetKRankApprox}
\begin{algorithmic}[1]
\Procedure{L1-KRankApprox}{$X,n,d,k$}
\State $r = O(k \log k)$
\State $m = O(r \log r)$
\State $t_1 = O(r \log r)$
\State $t_2 = O(m \log m)$
\State Generate all diagonal $R \in \R^{d \times d}$ with only $r$  $1$s 
\State Compute all possible sampling and rescaling matrices $D, T_1 \in \R^{n \times n}$ corresponding to Lewis Weights of $AR$ whose entries are powers of 2 between 1 and $1/nd$.  There are $m$ and $t_1$ nonzero entries on the diagonal, respectively.
\State Compute all sampling and rescaling matrices $T_2^T \in \R^{d \times d}$ according to the Lewis weights of $(DA)^T$ with $t_2$ nonzero entries, powers of 2 between 1 and $1/nd$ on the diagonal.
\State Evaluate $\Vert T_1 ARXYDAT_2 - T_1 A T_2 \Vert_1$ for all choices of above matrices.
\State Take the minimal solution
\State \Return $ARX, YDA$
\EndProcedure
\end{algorithmic}
\end{algorithm}

\begin{Lemma} \label{Runtime}
Algorithm \ref{DetKRankApprox} runs in time $\text{poly}(nd)$.
\end{Lemma}

\begin{proof}
Every matrix which is generated in Algorithm \ref{DetKRankApprox} has a number of nonzero entries bounded by $O(k \text{polylog} (k))$.  We can test all of the matrices which will take time proportional to the dimension of the matrix ($n$ or $d$) with exponent $O(k \text{polylog} (k))$ resulting in time $\text{poly}(nd)$ overall, since $k$ is constant. 
\end{proof}
We need one further lemma which describes the approximation error induced by using well-conditioned bases to decompose a matrix.

\begin{Lemma} \label{l1polyapprox}
Let $M \in \R^{N \times D}$ have rank $\rho$ and suppose $U \in \R^{N \times \rho}$ is a well-conditioned basis for $M$.  Let $M = US$ for a change of basis $S \in \R^{\rho \times D}$.  Then for all $x \in \R^D$:

\begin{equation*}
    \frac{\Vert Sx \Vert_1}{\text{poly}(D)} \le \Vert Mx \Vert_1 \le \text{poly}(D) \Vert Sx \Vert_1.
\end{equation*}
\end{Lemma}
\begin{proof}
For the left-hand side we can just calculate:
%
        \begin{align}
            \Vert Sx \Vert_1 &\le D \cdot \Vert Sx \Vert_\infty \\
                             &\le D \cdot \text{poly}(D) \Vert US x \Vert_1 \\
                             &= \text{poly}(D) \cdot \Vert Mx \Vert_1. 
        \end{align}
The second inequality follows from a property of the well-conditioned basis $U$.  The result follows from observing:
%
        \begin{align}
        \Vert Mx \Vert_1 = \Vert USx \Vert_1 &\le \Vert U \Vert_1 \Vert Sx \Vert_\infty \\
         &= \text{poly}(D) \Vert Sx \Vert_1
        \end{align}
\end{proof}

\subsection{Proof of Theorem \ref{l1Approxthm}}

For the proof of Theorem \ref{l1Approxthm} we introduce Algorithm \ref{DetL1Alg}.  It is enough to show that for every level, the low rank approximations of each group is polynomially bounded by $k$ in error.  The result follows by reasoning how this error grows as we progress through the tree.  Denote the $j$th block of $A$ by $A_{[j]}$. 

\begin{algorithm}
\caption{Deterministic $\ell_1$ low rank approx}  \label{DetL1Alg}
\begin{algorithmic}[1]
\Procedure{$\ell_1$-$k$-RankApprox}{$A,k, \gamma$}
\State $m,t \leftarrow 1, P_t \leftarrow 0$
\For{$i = 1: 1/\gamma$}
    \While{$m < n^{1 - i \gamma}$}
        \While{number of rows of $P_m < n^{\gamma}$}
            \State Run Algorithm \ref{DetKRankApprox} on $A_{[m]}$ and $k$ which outputs matrix $B \in \R^{k \times d}$
            \State $B \leftarrow WV^T$ ($k$-rank decomposition)
            \State Set $W = US$ for well-conditioned basis $U$
            \State $P_t \leftarrow SV^T$
            \State $m \leftarrow m + 1$
        \EndWhile
        \EndWhile       
        \EndFor
   \State Merge-and-Reduce all $P_m$ until we have an $n^{\gamma} \times d$ matrix. 
\State Set $P$ to be matrix of final $k$ rows.  
\State Solve $\min_Q \Vert QP - A \Vert_1 $.
\State \Return $QP$
\EndProcedure
\end{algorithmic}
\end{algorithm}

For every level in the tree we can take a group of rows, $C$, and perform Algorithm \ref{DetKRankApprox}.   For every $C$ used as input to Algorithm \ref{DetKRankApprox} a $k$-rank matrix $B$ of dimensions $n^\gamma \times d$ is returned. In particular, $B$ has the following property:

\begin{equation}
    \Vert C - B \Vert_1 \le \text{poly}(k) \min_{B' \text{rank} k} \Vert C - B' \Vert_1.
\end{equation}

Now factor $B$ using a $k$ rank decomposition.  That is, set $B = W
V^T$ where $W$ has $k$ columns and $V^T$ has $k$ rows. Further
decompose $W$ as $W = US$ for a well-conditioned basis $U$.  Note that
$W$ is $n^\gamma \times k$ (and of rank $k$) by the rank decomposition
so $U$ is also $n^\gamma \times k$ and $S$ is $k \times k$.  The dimensions of these matrices ensure that individually they do not exceed the space budget from the theorem.  

Apply Lemma \ref{l1polyapprox} with $W$ and $k$.  Then we have for every $x \in \R^k$ that $\Vert Sx \Vert_1 = \text{poly}(k) \Vert Wx \Vert_1$.  Since $U$ is $n^{\gamma}$ by $k$ and $k < \text{poly}(d)$, $U$ remains within the required space bound when we use it for the calculation.  Now ignore $U$ and store $SV^T$.  Note that each $SV^T$ is a matrix of $k$ directions in $\R^d$.  Pass $SV^T$ to the next level of the tree.

Merge the $SV^T$ for each group until we have a matrix of $n^\gamma$ rows.  Repeat the process over all $O(1/\gamma)$ levels in the tree. We require $n^\gamma d$ storage for every group so as we merge and pass $SV^T$ down the levels this combines to total storage of $O(1/\gamma) n^\gamma \text{poly}(d)$.  This part of the algorithm is a repeated use of Algorithm \ref{DetKRankApprox} which is $\text{poly}(nd)$ by Lemma \ref{Runtime} and some further lower time cost manipulations. Repeating these steps gives $\text{poly}(nd)$ as the overall time complexity. 

When this is done over all levels we will again have $k$ directions in $\R^d$.  Let $P$ be the matrix with these directions as rows.  Then we claim that $P$ can be used to construct our approximate $\ell_1$ low-rank approximation.

\begin{prop} \label{ApproxFactor}
Let $P$ be as described above. Then there exists $QP$ which is an $\ell_1$ low-rank approximation for $A$:

    \begin{equation*}
        \min_{Q} \Vert QP - A \Vert_1 \le \text{poly}(k) \Vert A - A' \Vert_1
    \end{equation*}

\end{prop}

\begin{proof}
    Each use of Algorithm \ref{DetKRankApprox} admits a $\text{poly}(k)$ approximation at every level of the tree.  Every time the well-conditioned basis $U$ is constructed and then ignored we admit a further $\text{poly}(k)$ error due to property 1 of Definition \ref{WCB}.  The distortion is blown up by a factor of $\text{poly}(k)$ every time we use Lemma \ref{l1polyapprox} which is at every level in the tree.
Hence, the total contribution of using Algorithm \ref{DetKRankApprox} is $\text{poly}(k)^{O(1/\gamma)}  = \text{poly}(k)$ for constant $\gamma$.  
\end{proof}

Proposition \ref{ApproxFactor} proves the approximation is $\text{poly}(k)$ as claimed.  By Lemma \ref{Runtime} we know that Algorithm \ref{DetKRankApprox} is $\text{poly}(nd)$ time.  The most costly steps in Algorithm \ref{DetL1Alg} are invocations of Algorithm \ref{DetKRankApprox} so combining this we see that the overall time cost is $\text{poly}(nd)$ as claimed, proving the theorem.

\section{Proofs for Section \ref{Deterministic ell_inf regression}} \label{sec: ell_inf proofs}

\begin{myproof}{Theorem}{\ref{ell_inf reg}}
{Given $A$, the first step is to store all rows of $A$ whose $\ell_p$ leverage score is above the threshold $\varepsilon/\text{poly}(d)$.  This step requires a polynomial increase to $\text{poly}(d)/ \varepsilon $ storage from Lemma \ref{lem: space_local_vs_global}.  Next the change of basis matrix $R$ is computed so that $AR$ is well-conditioned.  The stored matrix is $B$ with rows corresponding to those of large $\ell_p$ leverage scores from $A'$ and zero elsewhere.  Also, store all entries in $b$ whose magnitude is greater than $\varepsilon \Vert b \Vert_p$ and zero the rest out.  Call this vector $b'$.    

We now focus on the task of solving $\min_{x \in \R^d} \Vert A'Rx - b' \Vert_{\infty}$.  Any solution must necessarily have $\Vert x \Vert_p \le \text{poly}(d) \Vert b \Vert_p$ as otherwise $ x = \mathbf{0}$ is a better solution.    Recall that $\alpha = d^{1/p + 1/2}$ for a well-conditioned basis $AR$ with $p > 2$. Hence, the sum of all the $\ell_p$ leverage scores is $\alpha^p = d^{O(p)}$.  Then the number of rows with leverage score greater than the  $\varepsilon / \text{poly}(d)$ is at most $\text{poly}(d)/ \varepsilon \cdot d^{O(p)} = \text{poly}(d)/ \varepsilon$ for a constant $p$.  

Now, take any row for which the $\ell_p$ leverage score is less than the $\varepsilon / \text{poly}(d)$ threshold.
Then:
%
\begin{align*}
| \langle (AR)_i, x \rangle | & \le \| (AR)_i \|_{\infty} \| x \|_1 \\
        &\le \| (AR)_i \|_p \| x \|_1 \\
        &\le  \| (AR)_i \|_p \cdot d \| x \|_p \\
        &\le d \frac{\varepsilon}{\text{poly}(d)} \text{poly}(d) \| b \|_p.
\end{align*}

By an appropriate choice of the $\text{poly}(d)$ factors scaling
$\varepsilon$ we see that $| \langle (AR)_i, x \rangle | \le \varepsilon
\Vert b \Vert_p$.  On such coordinates the $\ell_{\infty}$ cost is
$|b_i| \pm \varepsilon \Vert b \Vert_p$ so by replacing the row with
one which is all zero we still pay $|b_i|$ which is within the
$\varepsilon \Vert b \Vert_p$ had we included the row.
The remaining high-leverage score rows are stored in their entirety so the cost on these rows is the same as in the original regression problem.}
\end{myproof}

\begin{myproof}{Theorem}{\ref{thm: ell_inf lower bound}}
{
Let $S$ be a set of $2^{\Omega(d)}$ strings in $\{0,1\}^d$ with each coordinate in a string uniformly sampled randomly from $\{0,1\}$.  Let $x,y \in S$ and fix a constant $0 < c < 1$.  By a Chernoff bound it follows that there are at least $cd$ coordinates in $[d]$ for which $x_i = 0$ and $y_i = 1$ with probability $1-2^{-\Omega(d)}$.  This implies for appropriate constants in the $\Omega(\cdot)$, by a union bound, all pairs of strings $x,y \in S$ have this property.  Hence, such an $S$ exists and we will fix this for the proof.

The regression problem can be reduced to an instance of the 
$\texttt{Indexing}$ problem \cite{knr99} in data streams as follows.
In the stream, the vector $b$ will be
all $1$s.  We will see a random subset $T$ of some elements from $S$.
We claim that it is possible to decide which case we are in: given a
random string $y$, whether $y$ is in $S$ independent of $T$, or $y$ is
in $T$.  This corresponds to solving $\texttt{Indexing}$ which
requires space $\Omega(|S|) = \Omega( \min \{n, 2^{\Omega(d)} \})$
even with randomization, via communication complexity arguments~\cite{Kushilevitz:Nisan:97}.

Given a test vector $y$, negate its coordinates so that $y \in \{0,-1\}^d$.  Now, append $y$ as a row to the final $b$ coordinate of $1$ at the end of the stream to obtain the last item in the stream $(y, 1)$.  If $y$ were in $S$ then both $y$ and its complement would be seen as rows of the matrix $A$.  Hence, the optimal cost for $\ell_{\infty}$-regression is at least 1.  Otherwise, $y$ is not in $S$.  Consider the set of coordinates $R$ where $y_i = 0$.  Set $x_i = 1/d$ for $i \in R$ and $-c/2d$ otherwise.  

Now we consider the cost of using $x$.  On the row corresponding to the negated vector $y$ the value will be at least $(-1) (-c/2d)(cd) = c^2/2$.  Since $b_i = 1$ the cost will be at most $|1 - c^2/2|$ for this coordinate.  On all other rows, by using the fact there are at least $cd$ occurrences of $x_i = 0, y_i = 1$ the value will be at least 

\begin{equation*}
cd(1/d) - (d-cd)(c/2d) \ge c - c/2 = c/2.
\end{equation*}
Hence the cost on these coordinates is at most $|1 - c/2 |$.  Since $c < 1$, the $\ell_{\infty}$ cost is at most $|1 - c^2/2 |$.  This is a constant factor less than the $\ell_{\infty}$ cost of 1 from the previous case so it is possible to decide which of the two cases we are in and hence the space is $\Omega( \min \{n, 2^{\Omega(d)} \})$ as claimed.
}
\end{myproof}

\section{Deterministic Approximate Matrix Multiplication} \label{sec: Mat Product}
\allowdisplaybreaks
Despite the generality of the subspace embedding result in Theorem \ref{DetEllpSEthm}, there may be occasions where the overheads are sufficiently large that it does not make sense to employ this method.  One such example is for the \emph{matrix multiplication} problem.  Let $A, B \in \R^{n \times d}$ and consider the task of finding a matrix $C$ for which $\Vert A^T B - C \Vert_1 < \varepsilon \Vert A \Vert_1 \Vert B \Vert_1$ where $0< \varepsilon < 1$ and the norm is \emph{entrywise $1$-norm}.

\begin{Lemma} \label{InnerProductLem}
Let $x,y \in \R^n$ have unit entrywise 1-norm.  Let $\varepsilon > 0$. Define:
\begin{alignat*}{2}
    \bar{x}_i = 
    & \begin{aligned} & \begin{cases}
  x_i & \text{ if } |x_ i| > \varepsilon/2, \\
  0  & \text{otherwise,}\\
  \end{cases}\\
  \end{aligned}
  & \hskip 3em &
  \begin{aligned}
  \bar{y}_i = 
  & \begin{cases}
  y_i & \text{ if } |y_ i| > \varepsilon/2, \\
  0  & \text{otherwise.}\\
  \end{cases} \\
  \end{aligned}
\end{alignat*}
Then $\langle x, y \rangle -  \varepsilon \le \langle \bar{x}, \bar{y} \rangle \le \langle x, y \rangle$ and this can be computed using space $O(1/ \varepsilon)$.
\end{Lemma}

\begin{proof}
Observe that $\bar{x}_i \le x_i$ and $\bar{y}_i \le y_i$ for $1 \le i \le n$.  Hence, $\langle \bar{x}, \bar{y} \rangle \le \langle x, y \rangle$.  For the left-hand side define the following sets:  $H_u = \{i : u_i > \varepsilon/2 \}, L_u = \{ i : u_i \le \varepsilon/2 \}$ for $u = x,y$.  Then we can write
\begin{equation*}
    \begin{aligned}
        \langle x,y \rangle &= \sum_{i \in H_x}x_i y_i + \sum_{i \in L_x}x_i y_i \\
        &= \sum_{\mathclap{i \in H_x \cap H_y}}x_i y_i + \sum_{\mathclap{i \in H_x \cap L_y}}x_i y_i + \sum_{\mathclap{i \in L_x \cap H_y}}x_i y_i + \sum_{\mathclap{i \in L_x \cap L_y}}x_i y_i \\
        &\le \langle \bar{x}, \bar{y} \rangle + \frac{\varepsilon}{2} \sum_{i \in H_x}x_i + \frac{\varepsilon}{2} \sum_{i \in H_y}y_i + \frac{\varepsilon}{2} \sum_{i \in L_y}y_i \\
      &\le \langle \bar{x}, \bar{y} \rangle + \varepsilon.
    \end{aligned}
\end{equation*}

Note that the sum can be written this way as the pair $H_x, L_x$ are disjoint, and likewise for $H_y, L_y$.  The first inequality follows from the second line because $i \in H_x \cap H_y$ means $x_i$ and $y_i$ are retained in $\bar{x}, \bar{y}$ so this summation corresponds directly to $\langle \bar{x}, \bar{y} \rangle$.  Then for every $i \in H_x \cap L_y$ we must have that $y_i \le \varepsilon/2$ so is bounded by $\frac{\varepsilon}{2} \sum_{i \in H_x}x_i$.  The same argument holds for the remaining two summations in the inequality. Finally, each of the three summations are at most 1 since both $x$ and $y$ have unit 1-norm.  The summations over $H_y$ and $L_y$ when combined are at most the norm of $y$ so can be combined such that $\sum_i y_i \le 1$.  This is enough to prove the result.
\end{proof}

The result for unit vectors is sufficient because we can simply normalize a vector, use Lemma \ref{InnerProductLem} and then rescale by the norm of $x$ and $y$.  This results in $\langle x, y \rangle - \varepsilon \Vert x \Vert_1 \Vert y \Vert_1  \le \langle \bar{x}, \bar{y} \rangle \le \langle x, y \rangle$.  This result can be used to prove the following theorem.

\begin{thm} \label{thm: matrix product 1}
Let $A, B \in \R^{n \times d}$ and let $\varepsilon > 0$.  Let $A_i$ denote the $i$th row of $A$ and $B^i$ denote the $i$th column of $B$.  For $X = A$ and $X = B$ define: 

\begin{equation*}
    \overline{X}_{ij} = 
     \begin{aligned}
         \begin{cases}
              X_{ij} & \text{ if } | X_{ij} | > \frac{\varepsilon}{2} \Vert X_i \Vert_1, \\
  0  & \text{otherwise.}
         \end{cases}
     \end{aligned}
\end{equation*}

Then in entrywise 1-norm: 

\begin{equation*}
\Vert AB^T - \overline{A} \overline{B}^T \Vert_1 \le  \varepsilon \Vert A \Vert_1 \Vert B \Vert_1.
\end{equation*}
\end{thm}

\begin{proof}
Fix $ \varepsilon > 0$.  The matrix product takes a row of $A$ with a column of $B^T$ which is simply a row of $B$.  These are both vectors in $\R^d$ so we can apply the transformation as in the Theorem statement, which is equivalent to that in Lemma \ref{InnerProductLem}.  By applying the rescaled version of Lemma \ref{InnerProductLem} we see that:

\begin{equation} \label{eq: mat prod}
| \langle A_i, B_j \rangle - \langle \overline{A}_i, \overline{B}_j \rangle \le \varepsilon \| A_i \|_1 \| B_j \|_1.
\end{equation}
\noindent Now the norm $\| AB^T - \overline{A} \overline{B}^T \|_1$ is the sum of all summands defined as in Equation \ref{eq: mat prod} over all pairs of $i$ and $j$.  Computing the sum then gives the desired result. 
%
\end{proof}

The argument from Lemma \ref{InnerProductLem} can easily be adapted to obtain a result for the matrix profuct $A^T B$.  
Observe that approximating $A^T B$ is equivalent to approximating inner products between columns of $A$ and columns of $B$.  
The modification is that the summary must be applied column-wise instead of row-wise as in Theorem \ref{thm: matrix product 1}.

\begin{thm} \label{thm: matrix product 2}
Let $A, B \in \R^{n \times d}$ and let $\varepsilon > 0$. Then there exists a deterministic algorithm which uses $O(1/\varepsilon)$ space and outputs $\overline{A}$ and $\overline{B}$ which satisfy:

\begin{equation*}
\| A^T B - \overline{A}^T \overline{B} \|_1 \le \varepsilon \|A\|_1 \|B\|_1.
\end{equation*}
\end{thm}

\begin{proof}
For a matrix $X$ let $X_i$ denote the $i$th row and $X^j$ denote the $j$th column.
Let $\| X_{:i}^j \|_1$ denote the 1-norm of column $j$ of $X$ up to and including row $i$.
It is clear that this norm is monotonic as more rows are seen in the stream.
In particular, $\| X_{:n}^j \|_1 = \| X \|_1$.
Therefore, the algorithm can be modified as follows: upon seeing a row $i$, if $|X_{ij}| > \varepsilon / 2 \cdot \| X_{:i}^j \|_1$ then keep the entry $X_{ij}$ and otherwise set $X_{oj} = 0$.
It is sufficient to consider only the last row. 
At this stage all rows which have not exceeded the running threshold upon seeing a particular row will have been ignored and only those which exceed $\varepsilon / 2 \cdot \| X_{:n-1}^j \|_1$ will be stored.
Then by increasing the threshold upon seeing row $n$ only the $X_{ij}$ which exceed $\| X_{:n}^j \|_1 = \| X \|_1$ will be kept and this is exactly the same set of rows as had the summary been applied given full access to the rows.

Hence, we may apply the result from Lemma \ref{InnerProductLem} on the columns of $A$ and $B$ as described above.  
It is then straightforward to show in a similar way to the lemma that the claim of the theorem holds.
\end{proof}

\section{Further Experimental results} \label{sec: experimental results appendices}
Here we illustrate the remaining experimental results on the YearPredictionMSD dataset
which include the space and time plots.
The experimental setup is the same as outline in Section \ref{sec: experiments}.

\begin{figure*}[t]
        \centering
        \subfloat[Summary Size]{
        \includegraphics[width=0.48\textwidth]{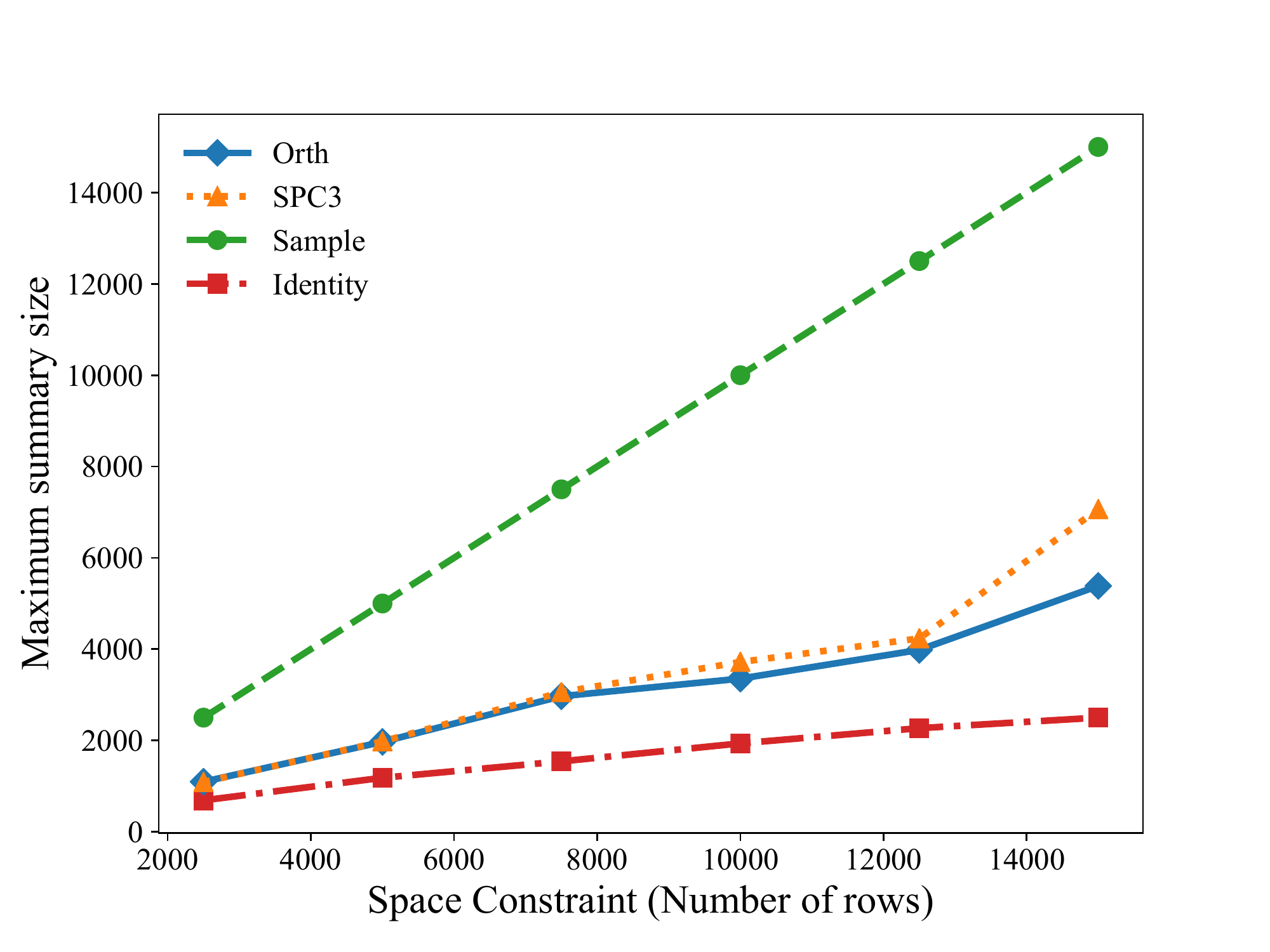}
        \label{fig: years_space_vs_block_size.pdf}}
        \hspace{\fill}
        \subfloat[Update Time]{
        \includegraphics[width=0.48\textwidth]{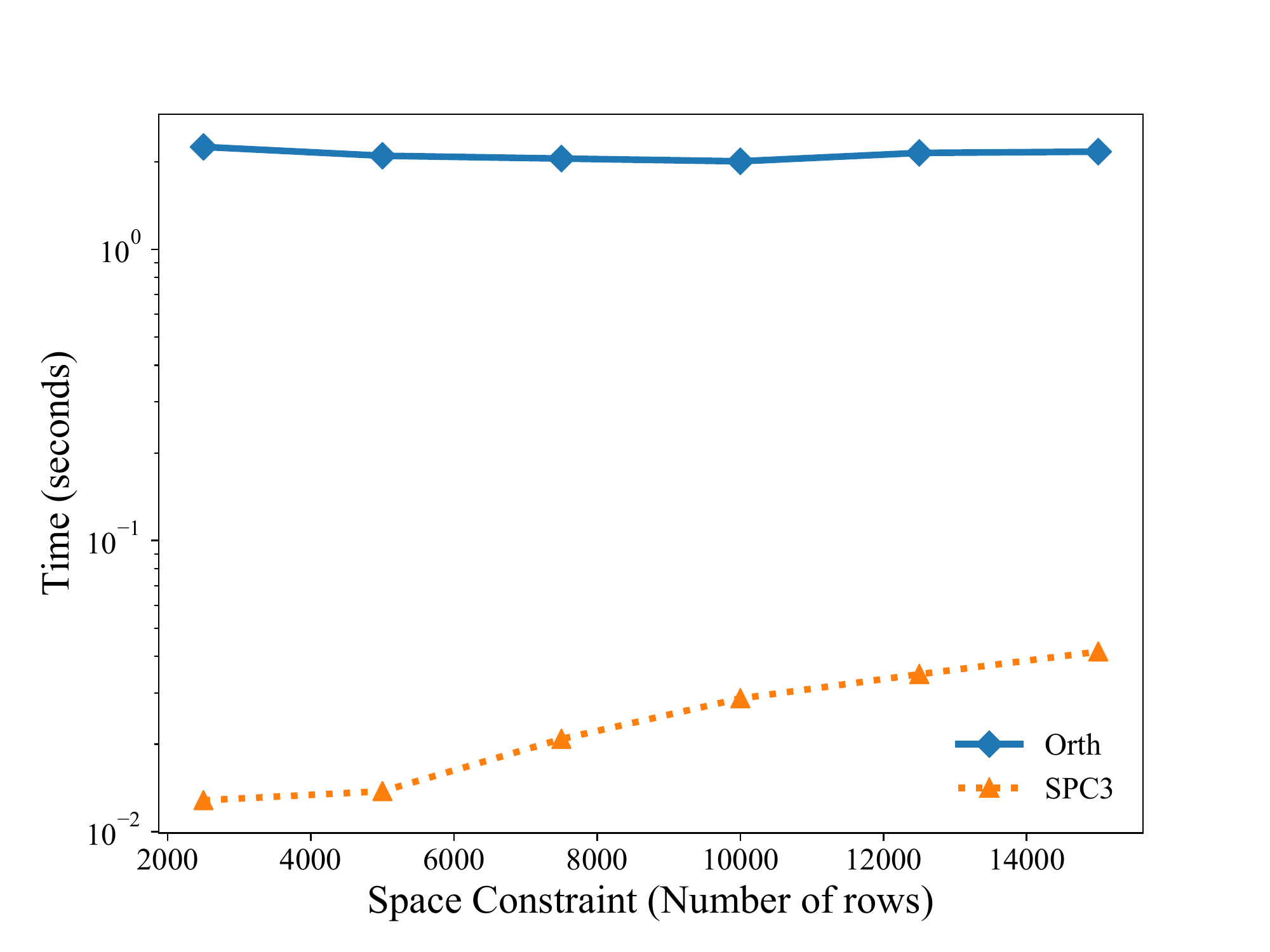}
        \label{fig: years_basis_time_vs_block_size.pdf}}
        \hspace{\fill}
        \subfloat[Query Time]{
        \includegraphics[width=0.48\textwidth]{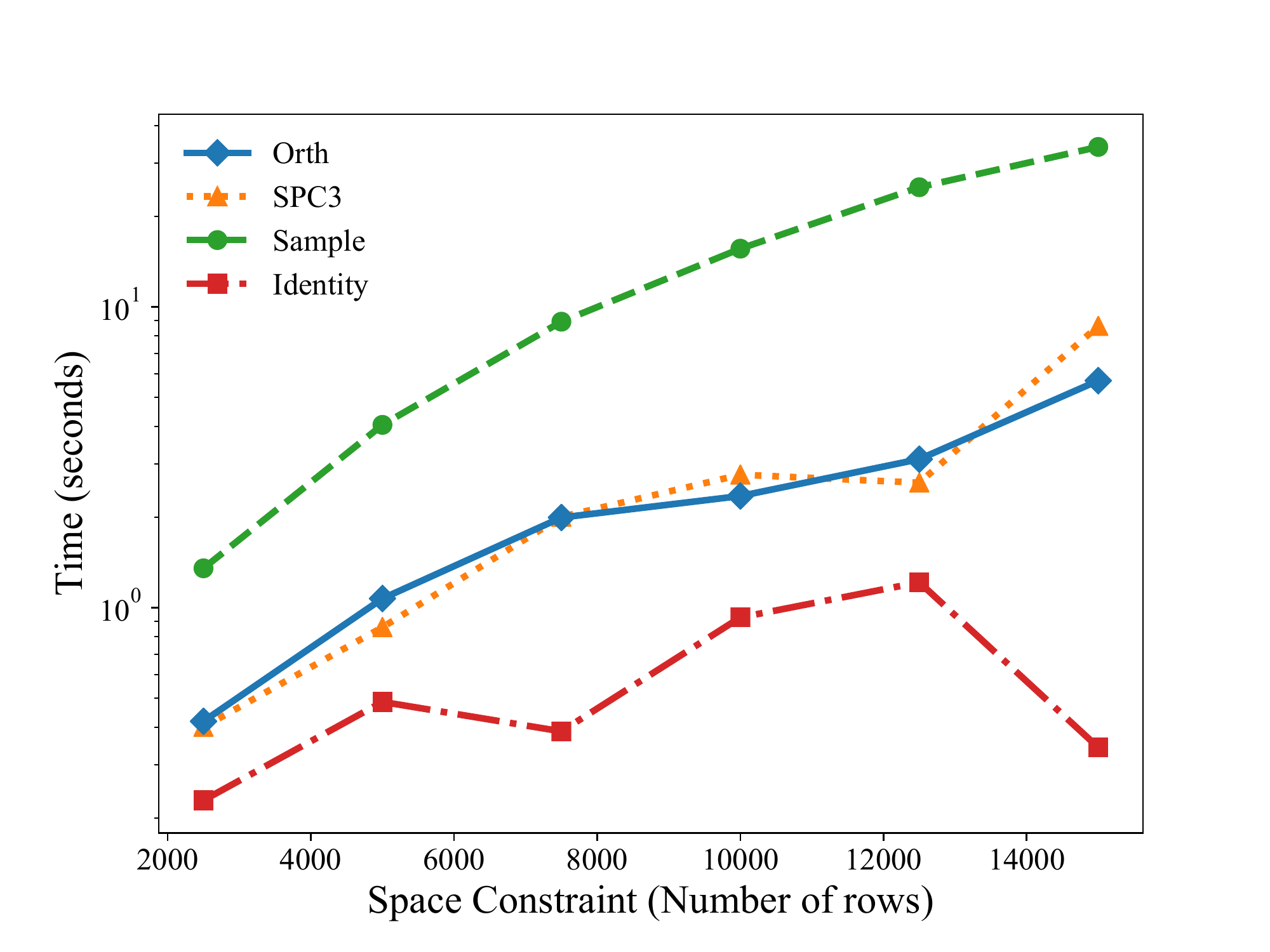}
        \label{fig: years_regression_time_vs_block_size.pdf}}
        \hspace{\fill}
        \subfloat[Total Time]{
        \includegraphics[width=0.48\textwidth]{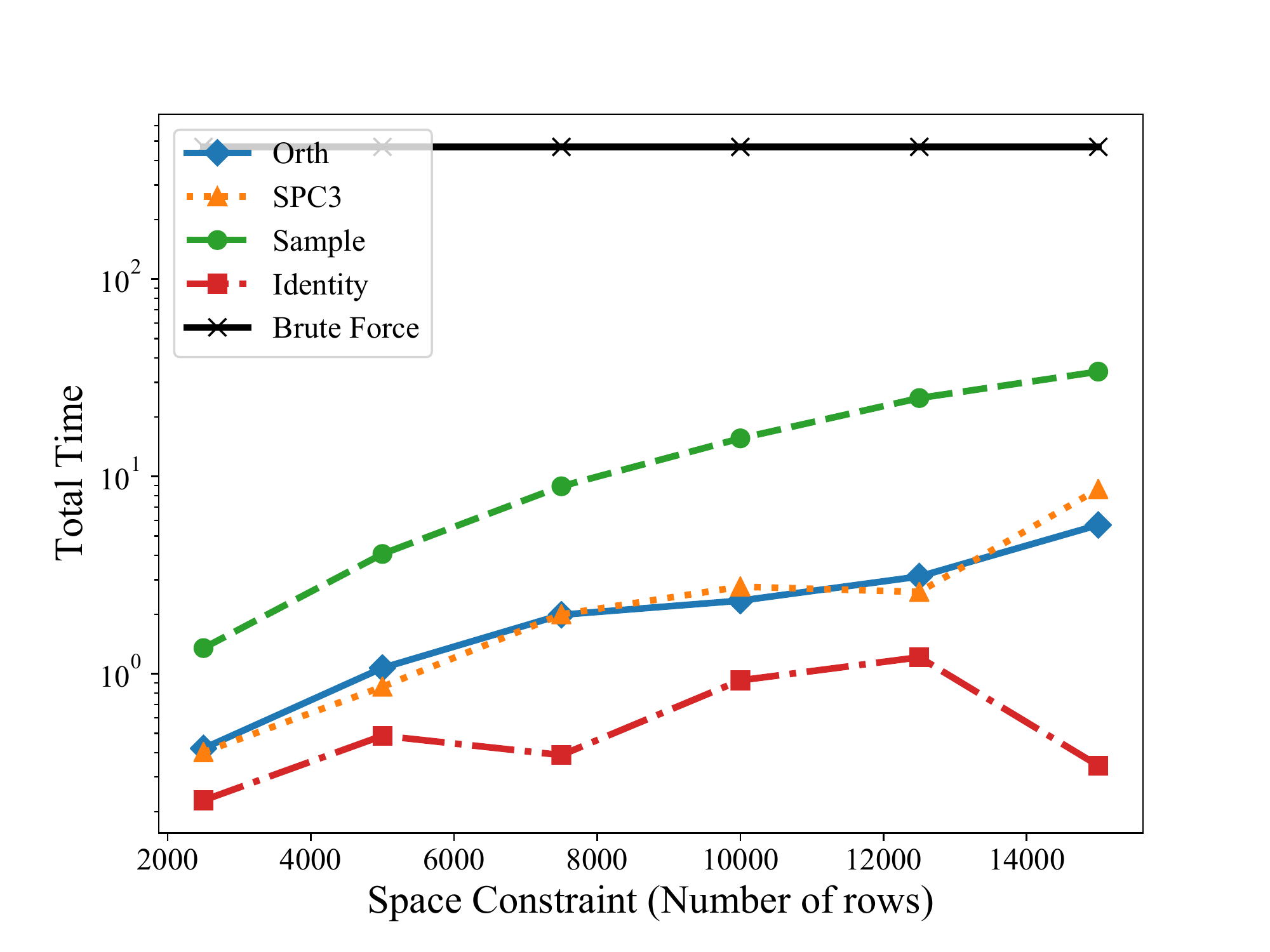}
        \label{fig: years_total_time_vs_block_size.pdf}}      
        
        \caption{Remaining plots for YearPredictionMSD data.}
\end{figure*}


\end{document}